\pdfoutput=1

\PassOptionsToPackage{most}{tcolorbox}
\newcommand{\tamulogofooter}{}
\documentclass[11pt,letterpaper]{mystyle}
\usepackage{fancyhdr}

\usepackage{tikz}
\usepackage[utf8]{inputenc} 
\usepackage[T1]{fontenc}
\usepackage[numbers]{natbib}
\usepackage{adjustbox}
\usepackage{breakurl}    
\usepackage{hyperref}
\usepackage[edges]{forest}
\usepackage{subcaption}
\usepackage{soul}
\usepackage{multirow}
\usepackage[utf8]{inputenc} 
\usepackage{booktabs}       
\usepackage{amsfonts}       
\usepackage{nicefrac}      
\usepackage{microtype}     
\usepackage{xcolor}        
\usepackage{colortbl}
\usepackage{enumitem}
\usepackage{amssymb}
\usepackage{wrapfig}
\usepackage{courier}
\usepackage{bxcoloremoji}
\usepackage{CJKutf8}
\usepackage{ragged2e}
\usepackage{longtable}
\usepackage[useregional]{datetime2}
\usepackage{threeparttable}
\usepackage{cleveref} 
\usepackage{wrapfig}
\usepackage{amsmath}
\usepackage{graphicx}  
\usepackage{amsmath}   
\usepackage{tcolorbox}
\tcbuselibrary{theorems}

\newtcbtheorem{proposition}{Proposition}{
    colback=gray!5,
    colframe=black!75,
    fonttitle=\bfseries
}{prop}

\crefname{appsec}{Appendix}{Appendices}
\Crefname{appsec}{Appendix}{Appendices}

\definecolor{tamuMaroon}{HTML}{500000}
\colorlet{tamuGrayMaroon}{tamuMaroon!15!black}
\newcommand{\github}{\raisebox{-1.5pt}{\includegraphics[height=1.05em]{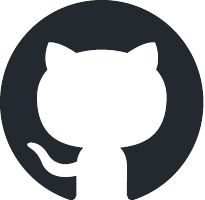}}}
\renewcommand{\today}{\DTMtoday}

\fancypagestyle{headstyle}{
    \fancyhead[C]{
    Equivariant Flow Matching for Electron Density Prediction
    }
}

\definecolor{hidden-red}{RGB}{205, 44, 36}
\definecolor{hidden-blue}{RGB}{194,232,247}
\definecolor{hidden-orange}{RGB}{243,202,120}
\definecolor{hidden-green}{RGB}{34,139,34}
\definecolor{hidden-pink}{RGB}{255,245,247}
\definecolor{hidden-black}{RGB}{20,68,106}
\definecolor{purple}{RGB}{144,153,196}
\definecolor{yellow}{RGB}{255,228,123}
\definecolor{hidden-yellow}{RGB}{255,248,203}
\definecolor{tkcolor}{RGB}{224,223,255}
\definecolor{darkblue}{rgb}{0, 0.40, 0.75}

\definecolor{tamuMaroon}{HTML}{500000}
\colorlet{abstractTextColor}{tamuMaroon!15!black}
\newcommand{\abstractstyle}{\color{abstractTextColor}}

\makeatletter
\let\oldabstract\abstract
\let\oldendabstract\endabstract
\renewcommand{\abstract}{\oldabstract\abstractstyle}
\renewcommand{\endabstract}{\oldendabstract}
\makeatother

\hypersetup{colorlinks=true, citecolor=darkblue, linkcolor=darkblue, urlcolor=darkblue}
\tcbset{
  aibox/.style={
    width=\linewidth,
    top=8pt,
    bottom=4pt,
    colback=blue!6!white,
    colframe=black,
    colbacktitle=black,
    enhanced,
    center,
    attach boxed title to top left={yshift=-0.1in,xshift=0.15in},
    boxed title style={boxrule=0pt,colframe=white,},
  }
}

\newtcolorbox{TakeawayBox}[2][]{takeawaybox,title=#2,#1}

\title{Equivariant Flow Matching for Electron Density Prediction}

\author{
  Chenxing Liang$^{1}$, 
  Chengdong Wang$^{2}$,
  Yuchao Lin$^{1}$, 
  \textbf{Xiaofeng Qian}$^{2,3,4*}$, 
  \textbf{Shuiwang Ji}$^{1,2,5*}$
\\

\vspace{1mm}

\normalfont{
$^1$ Department of Computer Science and Engineering, Texas A\&M University\vspace{-5pt} \\
$^2$ Department of Materials Science and Engineering, Texas A\&M University\vspace{-5pt} \\
$^3$ Department of Electrical and Computer Engineering, Texas A\&M University\vspace{-5pt} \\
$^4$ Department of Physics and Astronomy, Texas A\&M University\vspace{-5pt} \\
$^5$ Department of Mechanical Engineering, Texas A\&M University\vspace{-5pt} \\
}}

\begin{document}

\begin{abstract}

  \textbf{\large Abstract:}
  \vspace{1mm}

Machine learning surrogates for density functional theory (DFT) have been increasingly used to reduce the cost of first-principles calculations. In this arena, predicting real-space electron densities offers a scalable and transferable initialization for self-consistent field (SCF) procedures. However, current methods face a clear dilemma. That is, grid-based architectures incur a high computational cost, while basis-set methods fail to capture the structural correlations inherent in the coefficient space. Here, we develop \emph{OrbFlow}, an $\mathrm{SE}(3)$-equivariant generative model that predicts Gaussian-type orbital (GTO) coefficients via flow matching. OrbFlow retains the efficiency of a compact atom-centered basis while replacing pointwise regression with a learned probability path over the full coefficient space. It is trained through a two-phase trajectory curriculum that mitigates discretization drift during numerical integration. OrbFlow achieves state-of-the-art accuracy on QM9, reducing density error by 13.6\% relative to the previous best model, and reduces error by 51\% to 63\% on every molecule of the MD benchmark relative to the strongest prior method sharing its basis. The predicted density also cuts SCF iterations by up to 68\% with zero-shot transfer to unseen exchange-correlation functionals and recovers dipole and quadrupole moments to within a few percent of DFT references without any SCF calculation.

  \vspace{5mm}

  

  \vspace{1mm}
  \textbf{Keywords}: Equivariant Flow Matching, Electron Density Prediction, Generative Models, Quantum Chemistry
  
  \vspace{6mm}







    \textbf {* Equal senior authorship}

  \vspace{3mm}


  \github{} \textbf{Github Repository}: 
\url{https://github.com/divelab/AIRS/tree/main/OpenDFT/OrbFlow}

  \coloremojicode{1F917} \textbf{Model Checkpoints}: 
\url{https://huggingface.co/divelab/OrbFlow} 
  
  \coloremojicode{1F917} \textbf{Dataset}: 
\url{https://huggingface.co/datasets/divelab/OrbFlow-data}  
 


    \vspace*{0.2in}

\end{abstract}

\maketitle

\pagestyle{headstyle}
\thispagestyle{empty}

\newpage
\vspace{2em}
\tableofcontents

\newpage

\section{Introduction}
\label{sec:intro}

Density functional theory (DFT) is a cornerstone of computational quantum chemistry, materials science, and solid-state physics~\citep{hohenberg1964,kohn1965,jain2016,kirkpatrick2021,marzari2021}. It maps the many-body Schr\"odinger equation onto many one-electron Kohn-Sham equations, in which the total energy is a functional of electron density rather than many-electron wavefunctions, enabling \textit{ab initio} evaluation of molecular energies and electronic structure. The iterative self-consistent field (SCF) calculation that solves the Kohn-Sham equations scales cubically ($\mathcal{O}(N^3)$) with the number of electrons $N$, making it a bottleneck for large molecular structures, amorphous systems, and long-timescale molecular dynamics.

Recently, rapid advances have been made in developing deep learning surrogates to accelerate or bypass DFT, most commonly by predicting electronic structure operators directly from atomic coordinates~\citep{zhang2023artificial}, such as Kohn-Sham Hamiltonian matrix $\mathbf{H}$~\citep{tang2024deeph-hybrid,li2022deeph1,gong2023deephe3,yu2023qnet,yu2026qhnetv2,zhong2023hamgnn,gu2024deeptb,zhouyin2025deeptbe3,li2025wanet,luo2025sphnet,kim2025qhflow} or density matrix $\mathbf{P}$~\citep{shao2023machine,hazra2024predicting,febrer2025graph2mat} in localized basis sets. However, these matrix-centric predictions face systemic scalability and transferability bottlenecks. $\mathbf{H}$ is numerically difficult to predict owing to its large, dense entrywise structure~\citep{liu2026towards}. Models trained on small clusters therefore often fail to generalize to larger, unseen geometries, sometimes causing numerical divergence in the downstream DFT SCF calculations. In contrast, the real-space electron density is a physically smooth, localized, and universally transferable scalar field, and a high-quality initial guess for it can consistently reduce the number of SCF iterations under same convergence threshold.

Existing density prediction methods generally fall into two paradigms~\citep{fu2024scdp}. Probe-based methods evaluate density on dense 3D grids and offer strong expressiveness but incur heavy overhead from cubic spatial scaling. In contrast, basis-set methods are more efficient but rely on pointwise coefficient regression that fails to capture the structural correlations of the coefficient space. We propose to
address this dilemma by flow matching (FM)~\citep{lipman2023flowmatching,tong2024flowmatching2}, which parameterizes a time-dependent vector field that defines efficient ordinary differential equation (ODE) trajectories and yields rapid convergence while capturing intricate structural dependencies. 
FM has already achieved state-of-the-art accuracy in predicting quantum mechanical operators such as Kohn-Sham Hamiltonians~\citep{kim2025qhflow}, but its capabilities have not yet been applied to other continuous physical fields such as Gaussian-type orbital (GTO) basis spaces.


Motivated by the above considerations, we introduce \emph{OrbFlow}, a geometry-conditioned, $\mathrm{SE}(3)$-equivariant flow matching framework operating within a compact GTO coefficient space. Rather than a rigid, one-shot regression mapping, OrbFlow models the distribution of optimal basis spaces by learning a probability path that deforms a structured noise prior into the exact target distribution. To eliminate discretization drift during ODE integration rollouts, we propose a two-phase \emph{trajectory curriculum training strategy}. The model is first initialized with an independent flow matching objective to build a stable global vector field. It is then dynamically fine-tuned via an end-to-end differentiable multi-step Euler rollout loss that directly suppresses numerical truncation errors. Figure~\ref{fig:overview} presents an overview of OrbFlow.

\begin{figure}[t]
    \centering
    \includegraphics[trim={39pt 320pt 47pt 126pt}, clip, width=\textwidth]{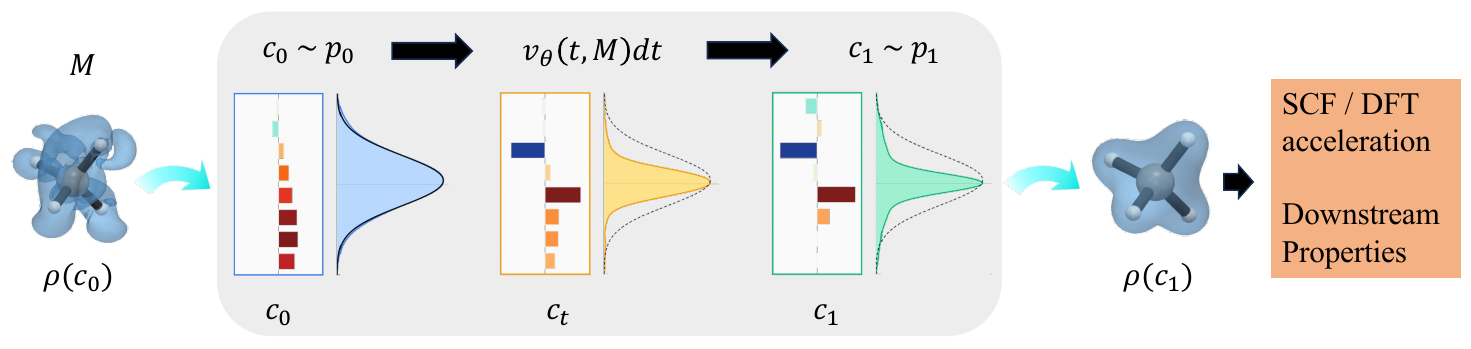}
    \caption{\textbf{Overview of OrbFlow.} Given a molecular graph $\mathcal{M}$, OrbFlow initializes Gaussian-type orbital (GTO) coefficients from a structured noise prior $\mathbf{c}_0 \sim p_0$ and integrates a learned, geometry-conditioned vector field $\mathbf{v}_\theta(t,M)$ over $t \in [0,1]$ to obtain the target coefficients $\mathbf{c}_1 \sim p_1$. The recovered coefficients are converted into a continuous real-space electron density $\rho(\mathbf{r})$ through the auxiliary GTO basis expansion, providing a fast surrogate for DFT initialization and enabling direct prediction of physical properties.}
    \label{fig:overview}
\end{figure}
\vspace{1pt}

Our primary contributions are summarized below.
\begin{itemize}[leftmargin=*]
    \item \textbf{Few-Step Flow Matching for Basis-Set Density Prediction.} We formulate GTO coefficient prediction as conditional flow matching, and show that its benefit is realized when the flow is trained under its own few-step rollout.
    \item \textbf{Trajectory Curriculum Training.} We introduce a two-phase curriculum that fine-tunes the flow under its own differentiable few-step rollout, requiring no teacher network and matching the accuracy of single-stage alternatives at $2.5\times$ less training compute.
    \item \textbf{Symmetry-Preserving Equivariant Integration.} We pair a high-order equivariant graph transformer with an irrep-isotropic prior, so that symmetry holds at every state along the rollout and the predicted density transforms equivariantly under rigid motions.
    \item \textbf{Accuracy and Downstream Utility.} OrbFlow achieves state-of-the-art accuracy on QM9 and reduces error on every MD molecule relative to the strongest prior model sharing its GTO basis, with SCF acceleration transferring zero-shot across unseen exchange-correlation functionals.
\end{itemize}

\section{Related Work}
\label{sec:related_work}

\textbf{Electron Density Prediction.} Continuous real-space electron density prediction generally falls into two paradigms~\citep{fu2024scdp}. Probe-based models evaluate density as spatial queries on dense 3D grids~\citep{jorgensen2022deepDFT,koker2024chargenet,gong2019probe1}, but cubic scaling with voxel resolution imposes severe memory and throughput costs at scale. Basis-set methods instead predict expansion coefficients of localized analytical functions, and split further by basis choice. Hybrid models combine a plane-wave basis with a localized GTO basis~\citep{kim2024gpwno}. Gaussian mixture models such as ELECTRA~\citep{elsborg2025electra} instead use a learnable set of Gaussian primitives rather than a fixed atom-centered basis. Auxiliary GTO regression predicts coefficients directly against a fixed atom-centered basis~\citep{cheng2023infgcn,fu2024scdp,qiao2022basis,del2023basis2}.

\textbf{Flow Matching and Exposure Bias.} Flow Matching (FM)~\citep{lipman2023flowmatching,tong2024flowmatching2} has emerged as a powerful alternative to diffusion models for molecular generative tasks~\citep{jing2024alphafold,miller2024flowmm,kim2025mofflow,stark2024docking}, most relevantly in QHFlow~\citep{kim2025qhflow}, which frames Hamiltonian matrix regression as an equivariant distribution-matching problem. We extend this paradigm to GTO coefficient spaces, pairing a flow over compact GTO coefficients with a trajectory-regularized multi-step Euler objective. The mismatch between single-step training and the compounding error of multi-step inference is known as exposure bias~\citep{bengio2015scheduled,ranzato2015sequence}. Diffusion and flow-based models remove it in two ways. Distillation trains a student to reproduce the multi-step behavior of a separately trained, frozen teacher, as in progressive distillation~\citep{salimans2022progressive} and consistency models~\citep{song2023consistency}. Shortcut models~\citep{frans2025one} drop the teacher instead, conditioning the network on the step size so that one model stays accurate at many step budgets. Our curriculum takes a third route and needs neither. It continues training the same network under its own $K$-step rollout, with identical weights across both phases, so it adds no model capacity, no teacher, and no extra conditioning input.

\section{Methods}
\label{sec:methods}

\subsection{Charge Density Representation via Gaussian-Type Orbitals}
\label{sec:representation}

To represent the continuous molecular electron density $\rho(\mathbf{r})$ compactly, we adopt an atom-centered Gaussian-type orbital (GTO) basis expansion~\citep{eichkorn1995auxiliary}, following \citet{fu2024scdp}. Each basis function $\varphi_{\alpha,\,\ell,\,m,\,\mathbf{R}_i}(\mathbf{r})$ is a normalized spherical Gaussian centered at node coordinate $\mathbf{R}_i$, characterized by radial exponent $\alpha$, angular momentum $\ell$, and magnetic quantum number $m$. The full functional form and the even-tempered scheme used to construct the radial exponents are given in Appendix~\ref{app:even-temper}. Given a molecular geometry $\{\mathbf{R}_i\}_{i=1}^{N}$ containing both nuclear coordinates and bond-midpoint virtual nodes, we parameterize the spatial density field as a linear combination of these localized GTOs,
\begin{equation}
  \rho(\mathbf{r})
  =
  \sum_{i=1}^{N}
  \sum_{j=1}^{N_i^{\mathrm{b}}}
  \sum_{m=-\ell_{i,j}}^{\ell_{i,j}}
  c_{i,j,m}\,
  \varphi_{\alpha_{i,j},\,\ell_{i,j},\,m,\,\mathbf{R}_i}(\mathbf{r}),
  \label{eq:rho-lcao}
\end{equation}
where $N_i^{\mathrm{b}}$ denotes the number of unique radial shells assigned to node $i$. Following the basis enrichment recipe of \citet{fu2024scdp}, we construct our basis by starting from uncontracted \texttt{def2-QZVPPD} elemental configurations~\citep{rappoport2010property}, embedding $\mathrm{SE}(3)$-equivariant virtual nodes at bond midpoints to capture interatomic shared densities, and applying even-tempered geometric shell augmentations~\citep{bardo1974even} for radial flexibility. Element-specific hyperparameters are detailed in Appendix~\ref{app:element-para}, and the representation floor of this basis, the error of the best fit it can express, is reported in Appendix~\ref{app:basis_floor}.

\subsection{Equivariant Flow Matching}
\label{sec:flow}

\textbf{Problem Formulation.} Our goal is to predict the target GTO coefficient
vector $\mathbf{c}^{\star} \in \mathbb{R}^{d}$ from a molecular graph
$\mathcal{G} = (\mathbf{R}, \mathbf{H})$, where $\mathbf{R} \in \mathbb{R}^{N \times 3}$
and $\mathbf{H} \in \mathbb{R}^{N \times D}$ denote the node coordinates and features. The density is reconstructed linearly from these coefficients on a discrete grid as
\begin{equation}
  \boldsymbol{\rho}(\mathbf{c}) = \mathcal{O}(\mathbf{c}) = s\,\boldsymbol{\Phi}\,\mathbf{c},
  \label{eq:synth}
\end{equation}
where $\Phi_{pj} = \varphi_j(\mathbf{r}_p)$ evaluates the $j$-th GTO function at grid point $p$, $s > 0$ is a scale factor, and $\mathcal{O}(\cdot)$ is our differentiable density synthesis operator. We train on $\mathcal{D} = \{(\mathcal{G}, \mathbf{c}^{\star})\}$, where target coefficients $\mathbf{c}^{\star}$ are obtained by projecting ground-truth DFT densities onto the basis via a regularized least-squares fit (Appendix~\ref{app:ridge}). Instead of mapping $\mathcal{G} \mapsto \mathbf{c}^{\star}$ in a single step, our model learns a transport from a tractable prior $\mathbf{c}_0 \sim p_0(\cdot \mid \mathcal{G})$ to the target $\mathbf{c}^{\star}$.

\textbf{Conditional Flow Matching in GTO Space.}
We formulate a continuous normalizing flow (CNF) within this compact, auxiliary GTO coefficient space: a time-dependent, geometry-conditioned vector field $\mathbf{v}_\theta(\mathbf{x}_t, t)$ defining an ODE $\mathrm{d}\mathbf{x}_t/\mathrm{d}t = \mathbf{v}_\theta(\mathbf{x}_t, t)$ over $t \in [0, 1]$ that deforms the prior $p_0(\mathbf{c}_0)$ at $t=0$ into a distribution on $\mathbf{c}^{\star}$ at $t=1$. Following \citet{kim2025qhflow}, we use a linear probability bridge,
\begin{equation}
  \mathbf{x}_t = (1-t)\,\mathbf{c}_0 + t\,\mathbf{c}^{\star},
  \label{eq:bridge}
\end{equation}
with conditional field $\mathbf{u}_t(\mathbf{x}_t \mid \mathbf{c}^{\star}, \mathbf{c}_0) = \mathbf{c}^{\star} - \mathbf{c}_0$. We parameterize $\mathbf{v}_\theta$ via an endpoint scheme,
\begin{equation}
  \mathbf{v}_\theta(\mathbf{x}_t, t) = \frac{f_\theta(\mathcal{G}, \mathbf{x}_t, t) - \mathbf{x}_t}{1-t},
  \label{eq:v-from-endpoint}
\end{equation}
where $f_\theta(\mathcal{G}, \mathbf{x}_t, t)$ predicts the clean target endpoint $\hat{\mathbf{c}}^{\star} \approx \mathbf{c}^{\star}$. Since our physical objective is real-space density accuracy rather than raw coefficient error, we supervise the predicted endpoint through the synthesis operator $\mathcal{O}(\cdot)$, evaluating the MAE over $P'$ subsampled spatial probes:
\begin{equation}
  \mathcal{L}_{\mathrm{stage1}}(\theta) = \mathbb{E}_{t, \,\mathcal{G},\,\mathbf{c}_0} \Bigl[ \tfrac{1}{P'} \sum_{p=1}^{P'} \bigl| \bigl[\mathcal{O}(f_\theta(\mathcal{G}, \mathbf{x}_t, t))\bigr]_{p} - \boldsymbol{\rho}^{\mathrm{lab}}_{p} \bigr| \Bigr],
  \label{eq:Lep}
\end{equation}

\textbf{Two-Phase Trajectory Curriculum Learning.} Training solely on isolated, random timesteps via~\eqref{eq:Lep} does not align the vector field with the multi-step trajectories used at inference. We address this with a two-phase curriculum. Phase 1 first establishes a good vector field at arbitrary, isolated timesteps. Phase 2 then fine-tunes that field under the same multi-step rollout used at inference, so training conditions match evaluation conditions directly.

\textit{Phase 1 (Field Initialization).} We minimize $\mathcal{L}_{\mathrm{stage1}}$ for $T_1$ steps, teaching the network to map arbitrary perturbed states $(\mathcal{G},\mathbf{x}_t,t)$ to a target endpoint whose synthesized density matches the true DFT field. This establishes a stable continuous normalizing flow without the overhead of unrolling the ODE, giving a stable warm-start vector field $\mathbf{v}_\theta$.

\textit{Phase 2 (Trajectory Regularization).} Starting from the Phase-1 weights, we continue training for $T_2$ steps under a full $K$-step Euler rollout that mirrors how the model is actually used at inference. Each rollout begins at the pure noise prior $\mathbf{x}_{t_0} = \mathbf{c}_0$ at $t_0 = 0$, rather than at an interpolated point along the bridge, and is unrolled end to end through the network in a differentiable manner. We discretize the interval $[0,1]$ into $K$ equal steps of size $\Delta t = 1/K$, with $t_k = k \cdot \Delta t$, and advance the state at each step with the Euler update
\begin{equation}
    \mathbf{x}_{t_{k+1}} = \mathbf{x}_{t_k} + \Delta t \cdot \mathbf{v}_\theta(\mathbf{x}_{t_k}, t_k), \quad \text{for } k = 0, \dots, K-1.
    \label{eq:generalized_euler}
\end{equation}
The final state $\hat{\mathbf{c}}_1 = \mathbf{x}_{t_K}$ is the model's simulated prediction of the target coefficients, using $K=2$ in our primary experiments. We train the model by comparing the density reconstructed from this simulated endpoint to the true DFT density, giving the Phase-2 objective
\begin{equation}
  \mathcal{L}_{\mathrm{stage2}}(\theta) = \mathbb{E}_{\mathcal{G},\,\mathbf{c}_0} \Bigl[ \tfrac{1}{P'} \sum_{p=1}^{P'} \bigl| \bigl[\mathcal{O}(\hat{\mathbf{c}}_1)\bigr]_{p} - \boldsymbol{\rho}^{\mathrm{lab}}_{p} \bigr| \Bigr].
  \label{eq:Lint}
\end{equation}

In Phase 1, the input state $\mathbf{x}_t$ is built by interpolating with the ground-truth target itself via~\eqref{eq:bridge}, so the network is partially shown the answer during training. In Phase 2, the rollout starts from pure noise with no such guidance, so the model must navigate the full trajectory on its own, and any small per-step error can accumulate into discretization drift. Phase 2 directly penalizes this accumulated drift. Penalizing this drift acts as a trajectory regularizer that straightens the flow field, which is why Phase 2 improves performance at the trained rollout horizon even though it starts from Phase-1 weights that already minimize $\mathcal{L}_{\mathrm{stage1}}$. Phase 2 achieves this by continuing to train the same network under its own rollout, rather than by introducing a separate teacher network or a distillation step.

\textbf{Equivariant Flow and Isotropic Irrep Priors.} For the generated density to be invariant under rigid translations and rotations across both phases, the flow must be strictly $\mathrm{SE}(3)$-equivariant. Atomic coordinates are centered at the origin ($\mathbf{r}_i \leftarrow \mathbf{r}_i - \tfrac{1}{N}\sum_j \mathbf{r}_j$), removing translational dependence. Under a rotation $R \in \mathrm{SO}(3)$, GTO basis functions (built from spherical harmonics $Y_\ell^m$) of angular momentum $\ell$ transform via the $(2\ell+1)$-dimensional Wigner D-matrix $D^{(\ell)}(R)$, so the coefficient vector $\mathbf{c}$ transforms covariantly as a direct sum of irreducible representations,
\begin{equation}
  \mathbf{c} \;\xrightarrow{\;R\;}\; \mathbf{D}(R)\mathbf{c} \;=\; \bigoplus_{\ell} D^{(\ell)}(R) \mathbf{c}^{(\ell)}.
\end{equation}
We therefore require both the network $f_\theta$ and the prior $p_0(\mathbf{c}_0)$ to respect this symmetry. We parameterize $f_\theta$ with a high-order equivariant graph transformer (Sec.~\ref{sec:impl}) satisfying
\begin{equation}
  f_\theta\bigl(R \cdot \mathcal{G}, \, \mathbf{D}(R)\mathbf{x}_t, \, t\bigr) = \mathbf{D}(R) f_\theta(\mathcal{G}, \mathbf{x}_t, t),
\end{equation}
and require $p_0(\mathbf{c}_0) = p_0(\mathbf{D}(R)\mathbf{c}_0)$ for all $R \in \mathrm{SO}(3)$, satisfied by an \emph{irrep-isotropic Gaussian prior} in which each angular momentum subspace $\ell$ is drawn independently from a zero-mean isotropic Gaussian, $\mathbf{c}_0^{(\ell)} \sim \mathcal{N}(\mathbf{0}, \sigma_\ell^2 \mathbf{I}_{2\ell+1})$, with $\sigma_\ell^2$ matched to the empirical variance of the target coefficients in subspace $\ell$.

\begin{proposition}{}
LLet $\mathbf{c}_0 \sim p_0(\mathbf{c}_0)$ be sampled from the irrep-isotropic prior. Under any rotation $R \in \mathrm{SO}(3)$, the prior density is invariant, $p_0(\mathbf{c}_0) = p_0(\mathbf{D}(R)\mathbf{c}_0)$. The proof is in Appendix~\ref{app:proof}.
\end{proposition}
Pairing this rotationally invariant prior with our equivariant network $f_\theta$ guarantees that integrating the velocity field yields a physically consistent, symmetry-preserving density.
\subsection{Model Implementation}
\label{sec:impl}

We parameterize $f_\theta$ using a stack of $L$ EquiformerV3 blocks~\citep{liao2026equiformerv3}, leveraging $\mathrm{SE}(3)$-equivariant multi-head attention and local $\mathrm{SO}(2)$ convolutions~\citep{passaro2023reducing} to jointly process the molecular geometry $\mathcal{G}$ and the current flow state $\mathbf{x}_t$. The backbone readout is mapped to per-atom GTO coefficients via a bilinear \texttt{e3nn} tensor-product readout~\citep{geiger2022e3nn}, serving directly as the flow head $f_\theta$ in~\eqref{eq:v-from-endpoint}. Full architectural details are in Appendix~\ref{app:training_details}.

\section{Experiments}
\label{sec:experiments}

We evaluate OrbFlow on the QM9 and MD benchmarks, building on the GTO auxiliary basis formulation of SCDP~\citep{fu2024scdp} and a state-of-the-art EquiformerV3~\citep{liao2026equiformerv3} backbone.  We evaluate predicted real-space densities $\boldsymbol{\rho}^{\mathrm{pred}} = \mathcal{O}(\hat{\mathbf{c}}_1)$ against ground-truth DFT references $\boldsymbol{\rho}^{\mathrm{lab}}$ on the full spatial grid. All QM9 and MD results are averaged over three independently trained random seeds. We report global reconstruction accuracy via the normalized Mean Absolute Error, $\mathrm{NMAE} = 100 \times \frac{\|\boldsymbol{\rho}^{\mathrm{pred}} - \boldsymbol{\rho}^{\mathrm{lab}}\|_{1}}{\|\boldsymbol{\rho}^{\mathrm{lab}}\|_{1}}$.

\subsection{QM9 Dataset}
\label{sec:exp-qm9}

\textbf{Dataset.} QM9~\citep{ruddigkeit2012qm91,ramakrishnan2014qm92} contains $\sim$134,000 stable organic molecules with up to nine heavy atoms ($\mathrm{C, O, N, F}$) plus $\mathrm{H}$. We use the real-space density files curated by \citet{jorgensen2022deepDFT} (VASP, PBE functional, PAW method), adopting the standard 123,835/50/10,000 train/val/test partition. To compare directly with ChargE3Net~\citep{koker2024chargenet}, InfGCN~\citep{cheng2023infgcn}, and GPWNO~\citep{kim2024gpwno}, we also evaluate on ``Split 2,'' the last 1,600 structures of the standard test set, matching these baselines' protocol.

\textbf{Results.} OrbFlow achieves state-of-the-art accuracy on both partitions (Table~\ref{tab:qm9-categorized}). It reaches 0.152\% NMAE on the primary test set, a 13.6\% reduction over the strongest baseline ELECTRA (0.176\%) and 14.6\% over its closest GTO-based predecessor SCDP + VO (0.178\%), and 0.165\% on Split 2, 13.6\% below SCDP + VO and 18.7\% below ChargE3Net (0.203\%). These margins reflect both the flow matching formulation and the EquiformerV3 backbone. Sec.~\ref{sec:ablations} separates the two at fixed architecture and compute, and finds that the flow matching formulation contributes on its own. They are also best read against the representation floor of GTO basis, since a direct least-squares fit of the ground-truth densities reaches 0.061\% NMAE (Appendix~\ref{app:basis_floor}) and OrbFlow removes 22\% of the error SCDP + VO leaves above it. Figure~\ref{fig:qualitative} visualizes predicted isosurfaces and spatial error.

\begin{table}[t]
\centering
\caption{QM9 charge density prediction error, grouped by model representation. The first column reports NMAE [\%] on the standard 10,000-molecule test set, and the second on ``Split 2,'' the last 1,600 structures of that test set. Results for OrbFlow and ELECTRA were obtained in this work as means over three training seeds, all other entries are quoted from the cited papers. Lower NMAE is better ($\downarrow$). \textbf{Bold} marks the best result per column, \underline{underline} the second-best.}
\label{tab:qm9-categorized}
\small
\renewcommand{\arraystretch}{1.15}
\begin{tabular}{@{}l cc@{}}
\toprule
\textbf{Model} & \textbf{NMAE [\%] $\downarrow$} & \textbf{NMAE, Split 2 [\%] $\downarrow$} \\
\midrule
\multicolumn{3}{@{}l}{\textit{Probe-based}} \\
\cmidrule(lr){1-3}
\hspace{0.6em}i-DeepDFT~\citep{jorgensen2022deepDFT} & $0.357 \pm 0.001$ & -- \\
\hspace{0.6em}e-DeepDFT~\citep{jorgensen2022deepDFT} & $0.284 \pm 0.001$ & -- \\
\hspace{0.6em}ChargE3Net~\citep{koker2024chargenet} & $0.196 \pm 0.001$ & $0.203 \pm 0.003$ \\
\cmidrule(lr){1-3}
\multicolumn{3}{@{}l}{\textit{Hybrid PW \& GTO}} \\
\cmidrule(lr){1-3}
\hspace{0.6em}GPWNO~\citep{kim2024gpwno} & -- & 0.73 \\
\cmidrule(lr){1-3}
\multicolumn{3}{@{}l}{\textit{Gaussian Mixture}} \\
\cmidrule(lr){1-3}
\hspace{0.6em}ELECTRA~\citep{elsborg2025electra} & \underline{$0.176 \pm 0.001$} & -- \\
\cmidrule(lr){1-3}
\multicolumn{3}{@{}l}{\textit{GTO-based}} \\
\cmidrule(lr){1-3}
\hspace{0.6em}InfGCN~\citep{cheng2023infgcn} & $0.869 \pm 0.002$ & 0.93 \\
\hspace{0.6em}SCDP, $L=3$~\citep{fu2024scdp} & $0.432 \pm 0.001$ & $0.438 \pm 0.003$ \\
\hspace{0.6em}SCDP + VO, $L=6$~\citep{fu2024scdp} & $0.178 \pm 0.001$ & \underline{$0.191 \pm 0.002$} \\
\hspace{0.6em}\textbf{OrbFlow (ours)} & $\mathbf{0.152 \pm 0.001}$ & $\mathbf{0.165 \pm 0.001}$ \\
\bottomrule
\end{tabular}
\end{table}

\textbf{Training Details.} We train with Adam under a linear warm-up followed by a power-decayed schedule, following the two-phase curriculum of Sec.~\ref{sec:flow} with rollout budget $K=2$. Full hyperparameters are in Appendix~\ref{app:dataset-settings}.

\subsection{MD Dataset}
\label{sec:exp-md}

\textbf{Dataset.} We benchmark on molecular dynamics (MD) trajectories spanning six molecules, ethanol, benzene, phenol, resorcinol, ethane, and malonaldehyde (MDA)~\citep{brockherde2017MD1,bogojeski2020MD2}, with reference densities computed at the PBE level via plane-wave Quantum ESPRESSO calculations, sampled across diverse geometries on uniform 3D grids.


\textbf{Results.} Table~\ref{table:md-test-performance} reports NMAE across the six MD trajectories. Compared with SCDP, the strongest prior GTO-based method, OrbFlow reduces test error on every molecule by 51\% to 63\% in relative terms, with the largest gains on benzene (1.13\% to 0.421\%) and ethane (2.05\% to 0.871\%). Averaged over the six molecules, NMAE falls from 1.81\% to 0.798\%, a $2.3\times$ reduction at comparable model size. On benzene and ethane OrbFlow comes within 0.02 and 0.06 percentage points of the basis fitting floor (Appendix~\ref{app:basis_floor}), leaving little room for further gains in this basis. Against ELECTRA, a 723M-parameter Gaussian-mixture model with roughly $55\times$ more parameters, OrbFlow remains competitive. OrbFlow is more accurate on benzene and ethane and statistically indistinguishable on phenol (0.575\% against 0.573\%, well within seed variance), while ELECTRA leads on ethanol, resorcinol, and malonaldehyde, by a small margin on the first two and a clear one on malonaldehyde. Therefore, OrbFlow approaches the accuracy of a much larger Gaussian-mixture model while improving substantially over every prior method operating in the same GTO basis, and it outperforms the hybrid plane-wave method GPWNO on all six molecules by wide margins. Figure~\ref{fig:qualitative} visualizes predicted density isosurfaces and spatial error for both benchmarks.

\begin{table}[t]
\centering
\caption{NMAE [\%] on MD test sets for OrbFlow, SCDP~\citep{fu2024scdp},
InfGCN~\citep{cheng2023infgcn}, ELECTRA~\citep{elsborg2025electra}, and GPWNO~\citep{kim2024gpwno},
grouped by model representation. \textbf{Bold} marks the best result per row,
\underline{underline} marks the second-best. Results for OrbFlow and ELECTRA were obtained in this work and are reported as mean $\pm$ std over three training seeds.
SCDP, InfGCN and GPWNO numbers are taken from~\citet{fu2024scdp}. ELECTRA's ethanol standard deviation is below $5\times10^{-4}$ and rounds to $0.000$.}
\label{table:md-test-performance}
\small
\setlength{\tabcolsep}{4pt}
\renewcommand{\arraystretch}{1.1}
\begin{tabular}{@{}l ccc c c@{}}
\toprule
\multirow{2}{*}{\textbf{Molecule}} & \multicolumn{3}{c}{\textit{GTO-based}} & \textit{Gaussian Mixture} & \textit{Hybrid PW \& GTO} \\
\cmidrule(lr){2-4} \cmidrule(lr){5-5} \cmidrule(lr){6-6}
& \textbf{OrbFlow (ours)} & \textbf{SCDP} & \textbf{InfGCN} & \textbf{ELECTRA} & \textbf{GPWNO} \\
\midrule
\# Params & 13M & 16M & -- & 723M & -- \\
\midrule
MD-ethanol       & \underline{$1.078 \pm 0.023$} & $2.34 \pm 0.25$ & 8.43  & $\mathbf{1.030 \pm 0.000}$ & 4.00 \\
MD-benzene       & $\mathbf{0.421 \pm 0.002}$ & $1.13 \pm 0.06$ & 5.11  & \underline{$0.472 \pm 0.013$} & 2.45 \\
MD-phenol        & $\mathbf{0.575 \pm 0.008}$ & \underline{$1.29 \pm 0.07$} & 5.51  & $\mathbf{0.573 \pm 0.012}$ & 2.68 \\
MD-resorcinol    & \underline{$0.663 \pm 0.014$} & $1.35 \pm 0.08$ & 5.95  & $\mathbf{0.619 \pm 0.005}$ & 2.73 \\
MD-ethane        & $\mathbf{0.871 \pm 0.008}$ & $2.05 \pm 0.12$ & 7.01  & \underline{$0.970 \pm 0.055$} & 3.67 \\
MD-malonaldehyde & \underline{$1.177 \pm 0.012$} & $2.71 \pm 0.60$ & 10.34 & $\mathbf{0.826 \pm 0.025}$ & 5.32 \\
\bottomrule
\end{tabular}
\end{table}

\textbf{Training Details.} Full details are in Appendix~\ref{app:dataset-settings}.


\begin{figure}[t]
\centering
\includegraphics[trim={75pt 117pt 119pt 116pt}, clip, width=\linewidth]{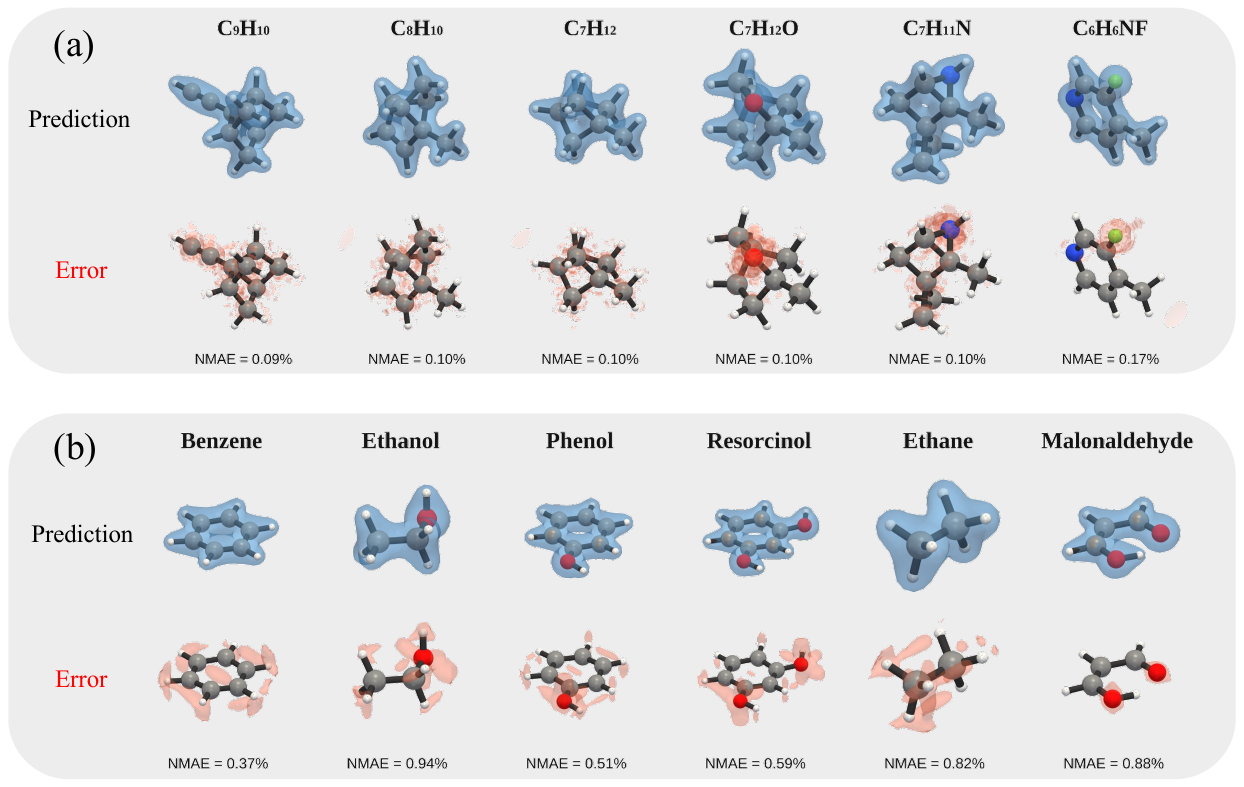}
\caption{Visualization of OrbFlow-predicted electron densities and prediction error.
(a)~Six representative QM9 test molecules.
(b)~Six representative MD test molecules.
\emph{Prediction} rows show the predicted density isosurface, rendered at
$\rho = 0.69$\,e/\AA$^3$ for QM9 and
$\rho = 1.12$\,e/\AA$^3$ for MD.
\emph{Error} rows show an isosurface of the absolute error
$|\rho^{\mathrm{pred}}-\rho^{\mathrm{DFT}}|$, at
$|\Delta\rho| = 1.2\times 10^{-3}$\,e/\AA$^3$ for QM9 and
$|\Delta\rho| = 1.1\times 10^{-2}$\,e/\AA$^3$ for MD.
Per-example NMAE is reported below for each molecule.}
\label{fig:qualitative}
\end{figure}

\subsection{DFT Acceleration and Functional Transferability}
\label{sec:exp-dft-acceleration}

A key downstream application of machine-learned charge density is to use it as improved initial guess for SCF calculations, thereby reducing the number of SCF iterations and accelerating DFT calculations. Standard Kohn-Sham DFT calculations typically initialize the SCF loop using a superposition of atomic densities (SAD), which may provide a relatively poor approximation to the converged electronic density and therefore require a larger number of computationally expensive SCF iterations. To improve the initialization, we instead use the real-space charge density $\rho_{\mathrm{pred}}$ predicted by OrbFlow as the initial density for the SCF calculation.

 The resulting SCF acceleration was evaluated using two complementary metrics. The Relative Iteration Count (RIC), defined as $\mathrm{RIC} = (N_{\mathrm{pred}} / N_{\mathrm{SAD}}) \times 100\%$, compares the number of SCF iterations required to reach the specified energy convergence criterion when initialized with the OrbFlow-predicted density and the conventional SAD density, respectively. Here, $N_{\mathrm{pred}}$ and $N_{\mathrm{SAD}}$ denote the corresponding numbers of SCF iterations. Similarly, the Relative Time Ratio (RTR), defined as $\mathrm{RTR} = (T_{\mathrm{pred}}/T_{\mathrm{SAD}}) \times 100\%$, compares the corresponding total wall-clock times, where $T_{\mathrm{pred}}$ and $T_{\mathrm{SAD}}$ denote the wall-clock times for calculations initialized with $\rho_{\mathrm{pred}}$ and SAD, respectively.

\textbf{Functional Transferability.} To assess whether OrbFlow learns transferable physical representations rather than memorizing artifacts specific to a single level of theory, we evaluate its zero-shot generalization across unseen exchange-correlation (XC) functionals using held-out QM9 molecules. All reference calculations are performed using VASP~\citep{dft_vasp1996prb} with a fixed plane-wave basis, with the complete computational settings provided in Appendix~\ref{app:vasp_settings}. Although trained exclusively at the GGA level, OrbFlow is evaluated across higher rungs of Jacob's ladder to hybrid functionals (B3LYP5~\cite{B3LYP5_Becke1993,B3LYP5_C.Lee1988}, PBE0~\cite{PBE0_Adamo1999,PBE0_Paier2005}), without any retraining. 

As shown in Table~\ref{tab:dft-transferability}, at the training-distribution functional PBE, OrbFlow reduces the required SCF iterations to 54.48\% of the SAD baseline and the wall-clock time to 89.51\%. The wall-clock saving is smaller because SCF iterations are only one component of a DFT run (Appendix~\ref{app:rtr_model}). This acceleration carries over to the two functionals never seen during training, both of which evaluate exact exchange rather than the semi-local exchange-correlation OrbFlow was trained on. PBE0 exhibits the strongest speedup on both metrics, reducing SCF iterations by 68.13\% and wall-clock time by 60.98\%, though its SAD baseline converges unusually slowly (43.02 iterations against 21.33 for PBE). B3LYP5 also benefits, with a larger wall-clock reduction than at the training-distribution functional (RTR 83.12\% against 89.51\%) though a higher relative iteration count (RIC 77.04\% against 54.48\%). Taken together, these results indicate that OrbFlow's predicted density is a genuinely informative initial guess across levels of theory rather than a narrow fit to PBE-specific dynamics.

\begin{table}[t]
\centering
\caption{SCF acceleration from OrbFlow initializations, evaluated on QM9 molecules. $N_{\mathrm{SAD}}$ and $N_{\mathrm{pred}}$ are SCF iterations to convergence from the SAD and OrbFlow initializations, respectively. Lower $\mathrm{RIC}\downarrow$ and lower $\mathrm{RTR}\downarrow$ both indicate greater acceleration. Iteration counts are means over the evaluated molecules and the $\pm$ values are standard deviations across those molecules. RIC and RTR are computed from the means.}
\label{tab:dft-transferability}
\small
\renewcommand{\arraystretch}{1.1}
\begin{tabular}{@{}l cc cc@{}}
\toprule
\textbf{Functional (Family: Name)} & \textbf{$N_{\mathrm{SAD}}$} & \textbf{$N_{\mathrm{pred}}$} & \textbf{RIC [\%]} $\downarrow$ & \textbf{RTR [\%]} $\downarrow$ \\
\midrule
\multicolumn{5}{@{}l}{\textit{Training-distribution functional}} \\
\hspace{0.6em}GGA: PBE & 21.33 $\pm$ 1.03 & 11.62 $\pm$ 1.65 & 54.48 & 89.51 \\
\midrule
\multicolumn{5}{@{}l}{\textit{Zero-shot transfer to unseen XC functionals}} \\
\hspace{0.6em}Hybrid GGA: B3LYP5 & 21.60 $\pm$ 1.73 & 16.64 $\pm$ 1.42 & 77.04 & 83.12 \\
\hspace{0.6em}Hybrid GGA: PBE0 & 43.02 $\pm$ 4.77 & 13.71 $\pm$ 2.47 & {31.87} & {39.02} \\
\bottomrule
\end{tabular}
\end{table}
\vspace{-4pt}

\subsection{Direct Physical Properties from Predicted Densities}
\label{sec:exp-direct-properties}

Beyond accelerating SCF, a predicted density also serves as a DFT surrogate. Density-derived observables follow directly from $\boldsymbol{\rho}^{\mathrm{pred}}$, with no self-consistent calculation. We evaluate the molecular dipole moment $\boldsymbol{\mu}$ and the traceless quadrupole tensor $Q_{ab}$ on 500 held-out QM9 molecules, against PBE references computed with VASP. Relative MAE is 1.0\% on the dipole magnitude and 1.1\% on the quadrupole norm, and 1.7\% and 2.5\% on their individual components (Table~\ref{tab:multipoles-from-density}, Appendix~\ref{app:multipole_def}). The accuracy is uniform across molecules rather than an average over well-predicted and poorly-predicted ones (Figure~\ref{fig:multipole-scatter}, Appendix~\ref{app:multipole_def}). Accuracy on both indicates that OrbFlow captures the density's shape at multiple spatial scales. One prediction therefore supports two downstream uses, an optimized SCF initial guess (Sec.~\ref{sec:exp-dft-acceleration}) and a stand-alone property surrogate.


\subsection{Inference Efficiency}
\label{sec:exp-efficiency}

Table~\ref{tab:efficiency-comparison} reports per-molecule inference latency. OrbFlow's end-to-end cost of $0.066$s is on par with ELECTRA ($0.061$s), and $1.6\times$ that of the single-shot GTO baseline SCDP, reflecting the two Euler steps of the rollout. The absolute overhead is a few hundredths of a second per molecule, negligible against a single SCF iteration, which typically costs on the order of seconds. Because the network cost scales close to linearly with $K$, OrbFlow's latency is predictable and directly controllable through the rollout budget. The full protocol is in Appendix~\ref{app:efficiency_details}.
\subsection{Ablation Studies}
\label{sec:ablations}

All ablations use a random 10\% subset of QM9, subsampled in both the training and test splits, held fixed across configurations, with three seeds and the protocol of Sec.~\ref{sec:exp-qm9}. Absolute NMAE is therefore higher than in Table~\ref{tab:qm9-categorized}, where models train on the full set. We ask three questions, whether the flow formulation contributes beyond the backbone, whether the two-phase curriculum is necessary, and how many Euler steps are needed.

\textbf{What the Flow Formulation Contributes.} Our flow formulation differs from direct regression in two ways. It starts from noise drawn from the prior rather than from a fixed input, and it is supervised on the output of its own multi-step rollout (\eqref{eq:Lint}) rather than one timestep at a time (\eqref{eq:Lep}). The four cells in Table~\ref{tab:ablation_factorial} run from direct regression, which does neither, to OrbFlow, which does both.

\begin{wraptable}[12]{r}{0.55\linewidth}
    \centering
    \vspace{-8pt}
    \caption{Ablations as a $2\times2$ over where the input state comes from and what the loss supervises. NMAE (\%), mean $\pm$ std over three seeds. All cells share the backbone, parameter count, basis, real-space loss, and step budget, and all use two network evaluations per prediction except direct regression, which uses one.}
    \label{tab:ablation_regression}\label{tab:ablation_factorial}
    \small
    \setlength{\tabcolsep}{3pt}
    \begin{tabular}{@{}lcc@{}}
        \toprule
        \textbf{Input state} & \textbf{Single-step loss} & \textbf{Rollout loss} \\
        \midrule
        Deterministic (from $\mathbf{0}$) & $0.280 \pm 0.003$ & $0.268 \pm 0.003$ \\
        Stochastic (noise prior) & $0.295 \pm 0.016$ & $\mathbf{0.240 \pm 0.005}$ \\
        \bottomrule
    \end{tabular}
\end{wraptable}

Neither ingredient works alone. Direct regression reaches $0.280 \pm 0.003\%$. Adding noise alone (Phase-1-only training) reaches $0.295 \pm 0.016\%$, slightly worse. Adding rollout alone, a deterministic two-step chain $\hat{\mathbf{c}} = g_\theta\big(\mathcal{G}, g_\theta(\mathcal{G}, \mathbf{0})\big)$, reaches $0.268 \pm 0.003\%$. Both together reach $0.240 \pm 0.005\%$. Starting from noise exposes the network to a wide range of input states, and rollout supervision makes that exposure count by training on the states inference actually visits, the model's own intermediates rather than points interpolated toward the target. All four cells share the backbone, parameter count, and $9\times10^4$-step budget, both rollout cells use two network evaluations, and the deterministic chain keeps a learned step-index embedding (Appendix~\ref{app:det-baseline}), so neither architecture, compute, nor step conditioning explains the gap. That gap is what few-step methods such as consistency~\citep{song2023consistency} and shortcut~\citep{frans2025one} models exist to close.

\begin{wraptable}[12]{r}{0.48\linewidth}
    \centering
    \vspace{-8pt}
    \caption{Curriculum variants on the same subset, with total training compute. The rollout-budget sweep and full compute accounting are in Appendix~\ref{app:curriculum_compute}.}
    \label{tab:ablation_curriculum}
    \small
    \setlength{\tabcolsep}{4pt}
    \begin{tabular}{@{}lcc@{}}
        \toprule
        \textbf{Configuration} & \textbf{NMAE (\%)} $\downarrow$ & \textbf{GPU-h} \\
        \midrule
        Phase 1 only & $0.295 \pm 0.016$ & 35 \\
        Phase 2 only & $0.250 \pm 0.027$ & 81 \\
        Joint single-stage & $0.241 \pm 0.002$ & 182 \\
        \textbf{Full curriculum (Ours)} & $\mathbf{0.240 \pm 0.005}$ & \textbf{74} \\
        \bottomrule
    \end{tabular}
\end{wraptable}

\textbf{Two-Phase Training Curriculum.} Both ingredients are needed, so the question becomes how to train them. Table~\ref{tab:ablation_curriculum} compares three routes to the same $K{=}2$ objective. Training directly at $K{=}2$ with no warm start reaches $0.250 \pm 0.027\%$, and optimizing the single-step and rollout losses jointly from random initialization reaches $0.241 \pm 0.002\%$, within seed variance of our $0.240 \pm 0.005\%$. Accuracy is therefore not what separates the rollout-trained variants. Cost is, 74 GPU-h for the curriculum against 81 without a warm start and 182 for joint optimization, and the curriculum reuses one Phase-1 checkpoint across rollout horizons while the alternatives retrain for each. Budgets here are matched in optimization steps rather than wall-clock, since a Phase-2 step evaluates the network twice.

\textbf{How Many Euler Steps.} We sweep the step count twice, over the Phase-2 rollout horizon $K$ used during training (Appendix~\ref{app:rollout_horizon}) and over the integration budget used at inference by the resulting $K{=}2$ model (Appendix~\ref{app:inference_budget}). In training, $K{=}1$ matches Phase-1-only (0.296\% against 0.295\%), since under the endpoint parameterization a single step from $t{=}0$ lands where the single-step loss already supervises, while $K{=}2$ and $K{=}3$ tie at a cost that grows with $K$. We therefore train at $K{=}2$. At inference, integrating more finely degrades accuracy, from $0.240\%$ at two steps to $0.539\%$ at four, because Phase 2 sharpens the network only at the states the two-step rollout visits. The train-matched budget is thus both the most accurate and the cheapest.

\section{Conclusion and Limitations}
\label{sec:conclusion}

We introduced OrbFlow, an $\mathrm{SE}(3)$-equivariant flow matching framework for electron density prediction. OrbFlow reaches state-of-the-art accuracy on QM9 and lowers error on every MD molecule against the strongest prior method sharing its basis, while matching a Gaussian-mixture model $55\times$ its size in accuracy and latency. Its predicted densities accelerate SCF across unseen exchange-correlation functionals and recover dipole and quadrupole moments without running SCF. The rollout adds a small inference overhead, and accuracy holds only at the trained horizon, so OrbFlow is best understood as a trained $K$-step map. Supervision across a range of horizons is a natural route to discretization-free behavior.

\section*{Acknowledgments}

This work was supported in part by the National Science Foundation under Grants IIS-2551110, MOMS-2331036, and CMMI-2226908; the Advanced Research Projects Agency for Health (ARPA-H) under grant 1AY1AX000053; the Texas A\&M University Division of Research Targeted Proposal Teams Funding Program; and the Texas A\&M Institute of Data Science Thematic Labs Program. First-principles calculations by C.W. were supported by the U.S. Department of Energy Office of Basic Energy Sciences under grant DE-SC0023866. We thank Jonas Elsborg and Arghya Bhowmik for their assistance in reproducing the results of the ELECTRA model. We are grateful to Shubham Parashar, Lianhao Zhou, and Wendi Yu for insightful discussions.

\bibliographystyle{refstyle}
\bibliography{ref}

\newpage
\appendix

\section{GTO Basis Set Details and Coefficient Projection}
\label{app:basis_details}
 
To train our machine learning model to predict the continuous real-space electron density $\rho(\mathbf{r})$, we project the continuous ground-truth density obtained from density functional theory (DFT) onto a localized, auxiliary basis set of Gaussian-Type Orbitals (GTOs). This section details the mathematical formulation of this basis set, the even-tempered scheme used to generate the radial exponents, and the regularized least-squares projection framework used to extract the target coefficient vectors.
 
\subsection{Even-Tempered Gaussian-Type Orbital (GTO) Basis Set}
\label{app:even-temper}
 
For an atom $i$ located at coordinate $\mathbf{R}_i$, the real-space density $\rho(\mathbf{r})$ is expanded using a linear combination of atom-centered GTO basis functions $\varphi_{\mu}(\mathbf{r} - \mathbf{R}_i)$. The basis functions are decomposed into radial and angular components,
\begin{equation}
    \varphi_{n, \ell, m}(\mathbf{r} - \mathbf{R}_i) = \mathcal{N}_{n, \ell} |\mathbf{r} - \mathbf{R}_i|^\ell \exp\left(-\alpha_{n, \ell} |\mathbf{r} - \mathbf{R}_i|^2\right) Y_{\ell}^{m}\left(\frac{\mathbf{r} - \mathbf{R}_i}{|\mathbf{r} - \mathbf{R}_i|}\right)
\end{equation}
where $Y_{\ell}^{m}$ represents the real spherical harmonics for angular momentum $\ell$ and magnetic quantum number $m \in [-\ell, \ell]$, and $\mathcal{N}_{n, \ell}$ is the normalization constant given by
\begin{equation}
    \mathcal{N}_{n, \ell} = \sqrt{\frac{2^{2\ell+7/2} \alpha_{n, \ell}^{\ell+3/2}}{(2\ell+1)!! \sqrt{\pi}}}
\end{equation}
where $(2\ell+1)!!$ denotes the double factorial of $2\ell+1$.
 
To systematically cover the radial coordinate space, the exponents $\alpha_{n, \ell}$ are parameterized using an even-tempered geometric progression. For each angular momentum channel $\ell$ up to $\ell_{\max}$, we construct $N_\ell$ radial shells. The exponents are determined by a base exponent $\alpha_{0, \ell}$ and a scaling factor $\beta_{\ell}$,
\begin{equation}
    \alpha_{n, \ell} = \alpha_{0, \ell} \cdot \beta_{\ell}^{n - 1}, \quad \text{for } n = 1, 2, \dots, N_\ell
\end{equation}
The hyperparameters $\{\alpha_{0, \ell}, \beta_{\ell}, N_\ell\}$ are optimized per element to span both the sharp, core-like densities near the nuclei (large $\alpha$) and the diffuse, bonding valence densities (small $\alpha$).

\subsection{Element-Specific Hyperparameter Configurations}
\label{app:element-para}
 
For each chemical element present in our dataset, Table~\ref{tab:gto_params} reports the maximum angular momentum $\ell_{\max}$, the number of radial shells $N_\ell$, the boundaries of the even-tempered progression ($\alpha_{\min} = \alpha_{0, \ell}$ and $\alpha_{\max} = \alpha_{0, \ell} \beta_\ell^{N_\ell-1}$), and the geometric progression scaling factor $\beta_\ell$, all instantiated from an uncontracted \texttt{def2-QZVPPD} parent basis with a shared even-tempered scale $\beta = 1.3$. QM9 uses the full table across $\mathrm{H, C, N, O, F}$. MD, whose six trajectories span only $\mathrm{H, C, O}$ (Sec.~\ref{sec:exp-md}), uses the corresponding subset. Bond-midpoint virtual nodes (Sec.~\ref{sec:representation}) are assigned the oxygen basis.
 
\begin{table}[htbp]
\centering
\caption{Even-tempered GTO parameters used for density projection (\texttt{def2-QZVPPD} parent basis, global $\beta=1.3$). $\alpha_{\max}=\alpha_{\min}\beta^{N_\ell-1}$. MD uses the H/C/O rows. QM9 uses the full table. Virtual nodes use O.}
\label{tab:gto_params}
\small
\setlength{\tabcolsep}{4pt}
\begin{tabular}{@{}lccccc@{}}
\toprule
\textbf{Element} & \textbf{Channel} & $N_\ell$ & $\alpha_{\min}$ & $\alpha_{\max}$ & $\beta$ \\
\midrule
H & $s$ ($\ell=0$) & 23 & 0.1606 & 51.59 & 1.3 \\
  & $p$ ($\ell=1$) & 18 & 0.1643 & 14.22 & 1.3 \\
  & $d$ ($\ell=2$) & 13 & 0.1681 & 3.917 & 1.3 \\
\midrule
C & $s$ ($\ell=0$) & 42 & 0.08996 & 4224 & 1.3 \\
  & $p$ ($\ell=1$) & 32 & 0.1169 & 398.2 & 1.3 \\
  & $d$ ($\ell=2$) & 27 & 0.1176 & 107.9 & 1.3 \\
  & $f$ ($\ell=3$) & 18 & 0.1529 & 13.22 & 1.3 \\
  & $g$ ($\ell=4$) & 13 & 0.1538 & 3.583 & 1.3 \\
\midrule
N & $s$ ($\ell=0$) & 42 & 0.1243 & 5837 & 1.3 \\
  & $p$ ($\ell=1$) & 32 & 0.1630 & 555.2 & 1.3 \\
  & $d$ ($\ell=2$) & 27 & 0.1641 & 150.5 & 1.3 \\
  & $f$ ($\ell=3$) & 18 & 0.2151 & 18.61 & 1.3 \\
  & $g$ ($\ell=4$) & 13 & 0.2165 & 5.044 & 1.3 \\
\midrule
O & $s$ ($\ell=0$) & 42 & 0.1393 & 6540 & 1.3 \\
  & $p$ ($\ell=1$) & 37 & 0.1102 & 1394 & 1.3 \\
  & $d$ ($\ell=2$) & 32 & 0.08720 & 297.0 & 1.3 \\
  & $f$ ($\ell=3$) & 23 & 0.1479 & 47.51 & 1.3 \\
  & $g$ ($\ell=4$) & 14 & 0.2509 & 7.600 & 1.3 \\
\midrule
F & $s$ ($\ell=0$) & 42 & 0.1753 & 8229 & 1.3 \\
  & $p$ ($\ell=1$) & 37 & 0.1280 & 1619 & 1.3 \\
  & $d$ ($\ell=2$) & 33 & 0.09354 & 414.2 & 1.3 \\
  & $f$ ($\ell=3$) & 23 & 0.1728 & 55.51 & 1.3 \\
  & $g$ ($\ell=4$) & 14 & 0.3194 & 9.673 & 1.3 \\
\bottomrule
\end{tabular}
\end{table}
 
 \subsection{L2 Least-Squares Coefficient Projection Mechanics}
\label{app:ridge}
 
Let $\varphi_p(\mathbf{r})$ represent a shorthand for the set of all active basis functions centered on all atoms (and virtual bond-midpoint nodes) in the system, indexed by a single flattened index $p = (i, n, \ell, m)$. Incorporating the scale factor $s > 0$ defined in the forward continuous density synthesis operator $\mathcal{O}(\mathbf{c}) = s\,\boldsymbol{\Phi}\,\mathbf{c}$, the continuous electron density $\rho(\mathbf{r})$ is approximated by the expansion
\begin{equation}
    \rho(\mathbf{r}) \approx \tilde{\rho}(\mathbf{r}) = s \sum_p c_p \varphi_p(\mathbf{r})
\end{equation}
 
To obtain the optimal target coefficients $c_p$ corresponding to the ground-truth DFT density $\rho_{\text{DFT}}(\mathbf{r})$, we minimize the $L^2$ error over the continuous space,
\begin{equation}
    E(\mathbf{c}) = \int_{\mathbb{R}^3} \left[ \rho_{\text{DFT}}(\mathbf{r}) - s \sum_p c_p \varphi_p(\mathbf{r}) \right]^2 d\mathbf{r}
\end{equation}
 
Setting the partial derivatives with respect to each coefficient to zero, $\frac{\partial E}{\partial c_p} = 0$, yields a system of linear equations, the normal equations,
\begin{equation}
    \sum_q S_{pq} c_q = y_p
\end{equation}
where $S_{pq}$ is the analytic overlap matrix between GTO basis functions,
\begin{equation}
    S_{pq} = \int_{\mathbb{R}^3} \varphi_p(\mathbf{r}) \varphi_q(\mathbf{r}) d\mathbf{r}
\end{equation}
and $y_p$ represents the normalized projection of the true DFT density onto the basis functions, factoring out the scale factor $s$,
\begin{equation}
    y_p = \frac{1}{s}\int_{\mathbb{R}^3} \rho_{\text{DFT}}(\mathbf{r}) \varphi_p(\mathbf{r}) d\mathbf{r}
\end{equation}
 
In practice, the ground-truth DFT reference density $\boldsymbol{\rho}^{\mathrm{lab}}$ is represented on a discrete numerical grid $\{(\mathbf{r}_k, w_k)\}$, where $\mathbf{r}_k$ is the coordinate of grid point $k$ and $w_k$ is the corresponding integration weight. Incorporating the diagonal spatial mask matrix $\mathbf{W} = \mathrm{diag}(w_1, w_2, \dots, w_P)$, the discrete least-squares objective formulated on the grid is expressed as
\begin{equation}
    \mathbf{c}^{\star} = \arg\min_{\mathbf{c}} \bigl\| \mathbf{W}^{1/2}\bigl(\boldsymbol{\Phi}\mathbf{c} - \mathbf{y}\bigr) \bigr\|_{2}^{2} + \lambda\|\mathbf{c}\|_{2}^{2} + \gamma\,\mathbf{c}^{\top}\! \mathrm{diag}(\boldsymbol{\Phi}^{\top}\boldsymbol{\Phi})\, \mathbf{c}
\end{equation}
where $\mathbf{y}=\boldsymbol{\rho}^{\mathrm{lab}}/s$, $\boldsymbol{\Phi}_{kj} = \varphi_j(\mathbf{r}_k)$ is the forward GTO evaluation design matrix, and $\lambda = 10^{-6}$ is a Tikhonov regularizer to prevent numerical instabilities caused by basis near-linear dependencies. The auxiliary parameter $\gamma$ regularizes the coefficients scaling with respect to the continuous basis functions' self-overlap.
 
Solving this regularized objective directly yields the target coefficients $\mathbf{c}^{\star}$,
\begin{equation}
    \mathbf{c}^{\star} = \left(\boldsymbol{\Phi}^{\top}\mathbf{W}\boldsymbol{\Phi} + \lambda \mathbf{I} + \gamma \, \mathrm{diag}(\boldsymbol{\Phi}^{\top}\boldsymbol{\Phi})\right)^{-1} \boldsymbol{\Phi}^{\top}\mathbf{W}\mathbf{y}
\end{equation}
During preprocessing we set $\lambda = 10^{-6}$ and $\gamma = 0$ (isotropic Tikhonov only). This is numerically stable while leaving a small, well-characterized representation residual.

\subsection{Empirical Fitting Floor}
\label{app:basis_floor}

Because OrbFlow predicts GTO coefficients rather than grid densities, learned NMAE cannot fall below the reconstruction error of $\mathbf{c}^{\star}$. We measure this floor by synthesizing $\tilde{\rho}=s\,\boldsymbol{\Phi}\,\mathbf{c}^{\star}$ from the stored ridge-fitted coefficients, with no neural network in the loop, and scoring NMAE on the same full probe grid used to evaluate the trained models.

\begin{table}[htbp]
\centering
\caption{GTO density fitting floor versus OrbFlow NMAE (\%). The floor is the NMAE of $\rho$ reconstructed from stored ridge ground-truth coefficients with no network. OrbFlow is mean $\pm$ std over three training seeds; the MD macro mean is unweighted over the six molecules.}
\label{tab:basis_floor}
\small
\renewcommand{\arraystretch}{1.1}
\begin{tabular}{@{}l r cc@{}}
\toprule
\textbf{Dataset} & \textbf{$n$} & \textbf{GT floor} $\downarrow$ & \textbf{OrbFlow} $\downarrow$ \\
\midrule
QM9 & 10{,}000 & $0.061$ & $0.152 \pm 0.001$ \\
\midrule
MD benzene & 500 & $0.399$ & $0.421 \pm 0.002$ \\
MD ethanol & 1{,}000 & $0.850$ & $1.078 \pm 0.023$ \\
MD phenol & 1{,}200 & $0.469$ & $0.575 \pm 0.008$ \\
MD resorcinol & 1{,}000 & $0.515$ & $0.663 \pm 0.014$ \\
MD ethane & 201 & $0.807$ & $0.871 \pm 0.008$ \\
MD malonaldehyde & 400 & $0.548$ & $1.177 \pm 0.012$ \\
MD macro mean & 6 & $0.598$ & $0.798$ \\
\bottomrule
\end{tabular}
\end{table}

On QM9 the fitting residual is $0.061\%$ while OrbFlow reaches $0.152\%$, so the remaining gap is due to coefficient prediction rather than incompleteness of the GTO span in Table~\ref{tab:gto_params}. Relative to SCDP + VO, which uses the same basis construction and reaches $0.178\%$, OrbFlow removes $22\%$ of the error remaining above this floor. On MD the gap varies by molecule: benzene and ethane lie close to the basis limit, whereas ethanol and malonaldehyde remain well above it, so compact rigid molecules are already well represented by this expansion while flexible oxygenates leave a larger residual for the model to close.
\subsection{Electron Count Conservation}
\label{app:electron_count}

Because a spherical Gaussian of angular momentum $\ell$ integrates to zero over $\mathbb{R}^3$ for every $\ell > 0$, only the $\ell{=}0$ channel of~\eqref{eq:rho-lcao} contributes to the total electron count,
\begin{equation}
  N_{\mathrm{e}}[\mathbf{c}]
  = \int_{\mathbb{R}^3} \rho(\mathbf{r})\,d\mathbf{r}
  = \sum_{i,j:\,\ell_{i,j}=0} c_{i,j,0}
    \int_{\mathbb{R}^3}
    \varphi_{\alpha_{i,j},0,0,\mathbf{R}_i}(\mathbf{r})\, d\mathbf{r},
  \label{eq:electron-count}
\end{equation}
so charge conservation is one linear condition on the $\ell{=}0$ coefficients rather than a property of the whole field.
We do not impose it during training, because a hard constraint would couple the equivariant readout across atoms.
When the predicted density is used as an SCF initial guess we renormalize it to the PAW valence count $N_{\mathrm{e}}^{\mathrm{val}}$.
On the official QM9 test set the closed-form count in~\eqref{eq:electron-count} for OrbFlow differs from $N_{\mathrm{e}}^{\mathrm{val}}$ by $0.014\,e$ on average, or $0.030\%$ of $N_{\mathrm{e}}^{\mathrm{val}}$.
The mean valence count on this set is $48.28\,e$, so the absolute miss is a few hundredths of an electron. The learned density already carries nearly the right charge, and the renormalization used for an SCF initial guess only has to correct a $0.030\%$ discrepancy.
\section{Proof of Prior Rotational Invariance}
\label{app:proof}
 
This section provides the full proof of the proposition stated in Sec.~\ref{sec:flow}, that the irrep-isotropic prior $p_0(\mathbf{c}_0)$ is invariant under any rotation $R \in \mathrm{SO}(3)$.
 
\begin{proof}
Because Wigner D-matrices $D^{(\ell)}(R)$ are unitary matrices (satisfying $D^{(\ell)}(R)^\top D^{(\ell)}(R) = \mathbf{I}$), the quadratic form in the Gaussian exponent of each subspace $\ell$ remains invariant under rotation,
\begin{equation}
  \|\mathbf{D}(R)\mathbf{c}_0^{(\ell)}\|^2 = \bigl(D^{(\ell)}(R)\mathbf{c}_0^{(\ell)}\bigr)^\top \bigl(D^{(\ell)}(R)\mathbf{c}_0^{(\ell)}\bigr) = \mathbf{c}_0^{(\ell)\top} D^{(\ell)}(R)^\top D^{(\ell)}(R) \mathbf{c}_0^{(\ell)} = \|\mathbf{c}_0^{(\ell)}\|^2.
\end{equation}
Thus, the density function is unchanged under any arbitrary 3D rotation, satisfying global $\mathrm{SO}(3)$-invariance.
\end{proof}

\section{Implementation Details}
\label{app:training_details}
 
This appendix consolidates the model architecture, target construction, and optimization procedure that support Sec.~\ref{sec:impl}. These settings are shared across QM9 and MD except where noted. Per-dataset hyperparameters and training configuration are given separately in Appendix~\ref{app:dataset-settings}.
 
\subsection{Network Architecture and Coefficient Parameterization}
\label{app:arch}
 
\textbf{Graph construction.} Each molecule is represented as a graph $\mathcal{G}$ whose nodes comprise nuclear coordinates together with bond-midpoint virtual orbitals (\texttt{vnode\_elem}${=}8$). Node features are initialized as $\ell{=}0$ spherical-tensor embeddings of atomic numbers, with higher-order channels set to zero. Edges are formed within a cutoff radius $r_{\mathrm{cut}}{=}6\,\text{\AA}$ and parameterized by Gaussian radial bases ($300$ functions, resolution $0.02\,\text{\AA}$) together with $\mathrm{SO}(3)$ edge frames used by equivariant attention.
 
\textbf{Equivariant backbone.} OrbFlow uses an EquiformerV3 encoder. QM9 employs $L{=}8$ blocks with $\ell_{\max}{=}6$, $m_{\max}{=}2$, and sphere/hidden/edge channel widths $128/256/128$. Each MD molecule uses a compact variant with $L{=}3$, $\ell_{\max}{=}3$, $m_{\max}{=}2$, and widths $96/192/96$ ($\sim$13M parameters). In both settings, attention uses $8$ heads with attention-hidden width $64$. No backbone dropout is applied.
 
\textbf{Flow conditioning and readout.} The scalar flow time $t$ is injected into the $\ell{=}0$ channel via a sinusoidal embedding. The continuous coefficient state $\mathbf{x}_t$ is introduced at each layer through an equivariant linear map, using EquiformerV3's SO(2)-convolution-based (eSCN-style) equivariant readout for coefficient entanglement~\citep{passaro2023reducing}, with no separate velocity MLP. This readout mechanism differs from using eSCN as a full alternative backbone. Latent tensors are mapped to per-atom GTO coefficients with multiplicities matched to the underlying basis. Coefficients are then scaled by $1/\sqrt{N_{\mathrm{orb}}}$ prior to density synthesis so that the training target remains consistent with the normalized coefficient representation.
 
\textbf{Basis and density synthesis.} We use an uncontracted \texttt{def2-QZVPPD} GTO basis with weighted even-tempered augmentation at ETB scale $\beta{=}1.3$ and orbital cutoff $5\,\text{\AA}$. Densities are synthesized in real space under the linear combination $\rho(\mathbf{r})=\sum_i c_i\,\phi_i(\mathbf{r})$.
 
\subsection{Flow-Matching Objective}
\label{app:objective}
 
Training follows conditional flow matching (CFM) with an \emph{endpoint} parameterization. The network predicts the terminal coefficient field, and the CFM velocity is implied by the linear bridge. 
The initial state is sampled from an irrep-structured Gaussian in normalized coefficient space, $c_0^{\mathrm{norm}} = \sigma\varepsilon + \mu$, where $\varepsilon$ is drawn via e3nn irrep noise and the scalar center $\mu$ is added only to the $\ell{=}0$ (scalar) coefficient dimension, following the standard convention for additive terms in equivariant architectures, since a nonzero mean is not rotation-invariant for any $\ell{>}0$ irrep. All $\ell{>}0$ channels remain zero-mean, consistent with the irrep-isotropic prior described in Sec.~\ref{sec:flow} and the invariance result of Proposition~1. On QM9 we use $(\mu,\sigma)=(2.5\times10^{-3},\,0.85)$ for the scalar channel. On MD, $(\mu,\sigma)$ is estimated once as the empirical mean and standard deviation of the scalar-channel training ground-truth coefficients for that molecule (Table~\ref{tab:md-priors}) and held fixed. Electron-count renormalization of the prior is disabled.

\begin{table}[htbp]
\centering
\caption{MD irrep-Gaussian prior parameters, estimated from each molecule's training-set ground-truth coefficients and held fixed during training.}
\label{tab:md-priors}
\small
\begin{tabular}{@{}lcc@{}}
\toprule
Molecule & $\mu$ & $\sigma$ \\
\midrule
Benzene        & $8.67\times10^{-3}$ & $1.82$ \\
Ethanol        & $5.64\times10^{-3}$ & $3.04$ \\
Phenol         & $5.04\times10^{-3}$ & $1.98$ \\
Resorcinol     & $4.72\times10^{-3}$ & $2.02$ \\
Ethane         & $8.44\times10^{-3}$ & $3.66$ \\
Malonaldehyde  & $1.79\times10^{-2}$ & $2.79$ \\
\bottomrule
\end{tabular}
\end{table}
 
\subsection{Optimization Protocol}
\label{app:optim}
 
All production runs use Adam~\citep{kingma2014adam} with $(\beta_1,\beta_2)=(0.9,0.999)$, $\varepsilon{=}10^{-8}$, base learning rate $1\times10^{-3}$, and zero weight decay. The learning rate follows a linear warm-up for the first $1{,}000$ steps (from a factor of $0.2$ to $1.0$), then a power decay
\begin{equation*}
  \eta_t \;=\; \eta_0\,\alpha^{t/\beta},
  \qquad \alpha{=}0.96,\quad \beta{=}4433.
\end{equation*}
We maintain an exponential moving average (EMA) of parameters with decay $0.995$, apply global-norm gradient clipping at $0.5$, and train in FP32 under PyTorch Lightning DDP. Checkpoints are selected by the lowest validation density error after Euler integration ($K{=}2$). All production jobs were run on NVIDIA H200 GPUs.
 
\section{Experiment Details}
\label{app:exp_details}
 
This appendix collects the protocols for the four supporting procedures referenced in the main text, the per-dataset training configuration and hyperparameters underlying every reported result, the DFT solver settings used for the SCF acceleration experiments (Sec.~\ref{sec:exp-dft-acceleration}), the inference latency measurement protocol underlying the computational efficiency comparison (Sec.~\ref{sec:exp-efficiency}), and the multipole moment definitions used to evaluate direct physical property prediction (Sec.~\ref{sec:exp-direct-properties}).
 
\subsection{Training Configuration}
\label{app:dataset-settings}

Table~\ref{tab:hparams} summarizes the primary architecture, coefficient representation, flow-matching, and optimization hyperparameters for both datasets. The training-time sampling range $t \sim \mathrm{Unif}[0.01, 0.99]$ excludes both endpoints of the bridge to avoid numerical instability during training, since the endpoint parameterization in Eq.~\ref{eq:v-from-endpoint} is singular at $t{=}1$ where the input state degenerates to $\mathbf{x}_t = \mathbf{c}^\star$, while at $t{=}0$ the input reduces to the pure noise prior and carries no target information.
 
\paragraph{QM9.} We train on the standard DeepDFT/SCDP partition ($123{,}835$ / $50$ / $10{,}000$ train/val/test) and report both the full test set and Split~2 (last $1{,}600$ test molecules) as in the main text. Training uses $P'{=}60{,}000$ randomly subsampled probe points per molecule ($120{,}000$ for validation/test), per-GPU batch size $4$ on $8$ GPUs (effective batch size $32$), Phase-1 length $T_1{=}7.5\times10^{4}$ steps, and Phase-2 continuation to a total of $T{=}4.5\times10^{5}$ steps.
 
\paragraph{MD.} We train a separate model for each of the six trajectories (benzene, ethanol, phenol, resorcinol, ethane, malonaldehyde), using the official per-molecule splits in Table~\ref{tab:md-splits}. Training, validation, and test all use $P'{=}125{,}000$ probes per configuration, per-GPU batch size $4$ on $4$ GPUs (effective batch size $16$), Phase-1 length $T_1{=}10^{4}$ steps, and Phase-2 continuation to a total of $T{=}1.1\times10^{5}$ steps.

\begin{table}[htbp]
\centering
\caption{Primary architecture and optimization hyperparameters. Shared entries are repeated for readability.}
\label{tab:hparams}
\small
\setlength{\tabcolsep}{4pt}
\renewcommand{\arraystretch}{1.12}
\begin{tabular}{@{}lcc@{}}
\toprule
\textbf{Hyperparameter} & \textbf{QM9} & \textbf{MD (per molecule)} \\
\midrule
\multicolumn{3}{@{}l}{\textit{EquiformerV3 backbone}} \\
\cmidrule(lr){1-3}
\# blocks $L$ & $8$ & $3$ \\
Sphere / hidden / edge channels & $128$ / $256$ / $128$ & $96$ / $192$ / $96$ \\
Attention hidden / \# heads & $64$ / $8$ & $64$ / $8$ \\
$\ell_{\max}$ ($m_{\max}$) & $6$ ($2$) & $3$ ($2$) \\
Cutoff $r_{\mathrm{cut}}$ / \# radial bases & $6\,\text{\AA}$ / $300$ & $6\,\text{\AA}$ / $300$ \\
Approx.\ \# parameters & $\sim$99M & $\sim$13M \\
\midrule
\multicolumn{3}{@{}l}{\textit{Coefficient / density representation}} \\
\cmidrule(lr){1-3}
GTO basis / ETB $\beta$ & \texttt{def2-QZVPPD} / $1.3$ & \texttt{def2-QZVPPD} / $1.3$ \\
GT projection ridge & $10^{-6}$ & $10^{-6}$ \\
Virtual orbitals & bond midpoints & bond midpoints \\
Probe subsample $P'$ (train / val) & $6{\times}10^{4}$ / $1.2{\times}10^{5}$ & $1.25{\times}10^{5}$ / $1.25{\times}10^{5}$ \\
\midrule
\multicolumn{3}{@{}l}{\textit{Flow matching}} \\
\cmidrule(lr){1-3}
Bridge / loss parameterization & CFM / endpoint & CFM / endpoint \\
Prior & irrep-Gaussian & irrep-Gaussian \\
Prior $(\mu,\sigma)$ & $(2.5{\times}10^{-3},\,0.85)$ & Table~\ref{tab:md-priors} \\
Train $t$ range / integrate grid & $[0.01,0.99]$ / uniform & $[0.01,0.99]$ / uniform \\
Euler steps $K$ (train / eval) & $2$ / $2$ & $2$ / $2$ \\
\midrule
\multicolumn{3}{@{}l}{\textit{Optimization}} \\
\cmidrule(lr){1-3}
Optimizer / base LR / weight decay & Adam / $10^{-3}$ / $0$ & Adam / $10^{-3}$ / $0$ \\
$(\beta_1,\beta_2)$ / warm-up steps & $(0.9,0.999)$ / $10^{3}$ & $(0.9,0.999)$ / $10^{3}$ \\
LR decay & power ($\alpha{=}0.96$, $\beta{=}4433$) & power ($\alpha{=}0.96$, $\beta{=}4433$) \\
EMA decay / grad clip / precision & $0.995$ / $0.5$ / FP32 & $0.995$ / $0.5$ / FP32 \\
Per-GPU batch / \# GPUs (eff.\ batch) & $4$ / $8$ ($32$) & $4$ / $4$ ($16$) \\
Phase 1 $\rightarrow$ total steps & $7.5{\times}10^{4}\rightarrow4.5{\times}10^{5}$ & $10^{4}\rightarrow1.1{\times}10^{5}$ \\
Random seeds & $0$, $42$, $123$ & $0$, $42$, $505$ \\
\bottomrule
\end{tabular}
\end{table}

\begin{table}[htbp]
\centering
\caption{MD trajectory split sizes used in our experiments.}
\label{tab:md-splits}
\small
\begin{tabular}{@{}lccc@{}}
\toprule
Molecule & Train & Validation & Test \\
\midrule
Benzene        & 1001 & 500  & 500  \\
Ethanol        & 1002 & 1000 & 1000 \\
Phenol         & 1001 & 1200 & 1200 \\
Resorcinol     & 1004 & 1000 & 1000 \\
Ethane         & 2000 & 201  & 201  \\
Malonaldehyde  & 2000 & 400  & 400  \\
\bottomrule
\end{tabular}
\end{table}

\subsection{DFT Solver Settings}
\label{app:vasp_settings}
 
This section provides full plane-wave DFT solver settings for the SCF acceleration experiments in Sec.~\ref{sec:exp-dft-acceleration}. For unseen XC functionals, the evaluation is performed on a fixed test set of 200 molecules. For training-distribution PBE, we expand the test set to 500 molecules because of the lower computational cost.
 
\textbf{DFT Calculations and Exchange-Correlation Functionals.} All calculations use VASP~\citep{dft_vasp1996prb} with the projector-augmented wave (PAW)~\citep{PAW_Bloch1994} method and the PBE exchange-correlation energy functional\citep{GGA_PBE_Perdew1996} (Sec.~\ref{sec:exp-dft-acceleration} additionally evaluates B3LYP5\citep{B3LYP5_Becke1993,B3LYP5_C.Lee1988} and PBE0~\citep{PBE0_Adamo1999,PBE0_Paier2005} for transferability). The kinetic energy for plane-wave basis set is fixed at $400$~eV, and the reciprocal-space sampling is restricted to the $\Gamma$-point only, consistent with the isolated-molecule (non-periodic) nature of the QM9 evaluation set. The blocked-Davidson algorithm was used for all SCF calculations except those employing B3LYP5, for which a conjugate gradient-based all-band simultaneous optimization scheme was adopted. This choice was made because blocked-Davidson algorithm exhibited poor SCF convergence for B3LYP5, with pronounced oscillations during electronic iterations that hindered convergence. The all-band optimization scheme provided more robust and stable convergence in these cases.
 
\textbf{Convergence Criterion.} SCF convergence is reached when the change in total energy between consecutive iterations falls below $10^{-5}$~eV, matching the convergence criterion used to report $\mathrm{RIC}$ and $\mathrm{RTR}$ in Sec.~\ref{sec:exp-dft-acceleration}.
 
\textbf{Cell Construction.} Molecules are placed in an orthorhombic simulation cell with a minimum of 2~\r{A} of vacuum padding on each side of the cell to prevent spurious periodic self-interaction, corresponding to a minimum distance of 4~\r{A} between the atoms from periodic images, same as the original charge density data set~\citep{jorgensen2022deepDFT}. The real-space FFT grid resolution is fixed at approximately 0.15~\r{A} per voxel along each cartesian direction across all systems, ensuring consistent grid density independent of molecule size.
 
\textbf{Initialization Protocol.} For each test molecule, we run two matched SCF calculations that differ only in initial charge density, an atomic superposition baseline (SAD initialization) and initialization from the OrbFlow-predicted density, renormalized to conserve total valence electron count and written directly as the solver's input charge density file. All other solver settings, including $k$-point sampling, energy cutoff, and convergence threshold, are held identical between the two runs to ensure that any difference in required SCF iterations is attributable solely to the quality of the initial density guess.

\textbf{Reported Metrics.} For each molecule, we record the number of SCF iterations required to reach convergence from each initialization and compute the Relative Iteration Count $\mathrm{RIC} = N_{\mathrm{pred}} / N_{\mathrm{SAD}} \times 100\%$. The Relative Time Ratio additionally accounts for OrbFlow's own inference latency, ensuring that reported acceleration reflects genuine end-to-end speedup rather than SCF iteration count alone.

\subsection{Why Iteration Savings Only Partially Translate to Wall-Clock Savings}
\label{app:rtr_model}

A DFT run costs more than its SCF iterations. Basis and PAW projector construction, FFT grid setup, symmetry analysis, and file I/O are incurred once per calculation and are unaffected by the quality of the initial density, so a reduction in iteration count can never reduce total runtime proportionally. Writing the total time of a run as a fixed overhead plus a per-iteration cost, $T = T_{\mathrm{fix}} + N\,t_{\mathrm{iter}}$, and letting $f = T_{\mathrm{fix}}/T_{\mathrm{SAD}}$ denote the non-iterative share of a baseline run, the two reported metrics are related by
\begin{equation}
  \mathrm{RTR} = f + (1-f)\,\mathrm{RIC}.
  \label{eq:rtr-ric}
\end{equation}
Solving~\eqref{eq:rtr-ric} for $f$ at each functional in Table~\ref{tab:dft-transferability} gives $f = 77.0\%$ for PBE, $26.5\%$ for B3LYP5, and $10.5\%$ for PBE0. Substituting these back reproduces all three reported RTR values exactly. PBE iterations on molecules of this size are cheap, so fixed setup dominates and even a 46\% reduction in iterations buys only about 10\% of wall-clock time. Hybrid functionals evaluate exact exchange at every iteration, which makes the iterative part far more expensive, shrinks $f$, and brings RTR close to RIC. Since $\mathrm{RTR} \to \mathrm{RIC}$ as $f \to 0$, the wall-clock benefit of a learned initial density grows with the cost of the underlying method, and the PBE row understates what the same iteration savings are worth at higher levels of theory or on larger systems. Note that $T_{\mathrm{pred}}$ includes OrbFlow's own inference latency, so these values of $f$ are mild overestimates.

\subsection{Multipole Moments}
\label{app:multipole_def}
 
This section gives the full definitions of the dipole and quadrupole moments evaluated in Sec.~\ref{sec:exp-direct-properties}, and supplements the aggregate accuracy reported there (Table~\ref{tab:multipoles-from-density}) with a per-molecule view.
 
Because the PAW reference densities contain only valence electrons, we pair them with the valence (pseudo-ion) charges $Z_A^{\mathrm{val}}$, so that the total system is charge neutral. The molecular electric dipole moment is then obtained via
\begin{equation}
  \boldsymbol{\mu}
  =
  \sum_{A} Z_A^{\mathrm{val}} \mathbf{R}_A
  -
  \int \mathbf{r}\, \boldsymbol{\rho}(\mathbf{r})\,d\mathbf{r},
  \label{eq:dipole}
\end{equation}
and the traceless electric quadrupole moment tensor is obtained analogously via
\begin{equation}
  Q_{ab}
  =
  \sum_{A} Z_A^{\mathrm{val}} \left(3 R_{A,a} R_{A,b} - |\mathbf{R}_A|^2 \delta_{ab}\right)
  -
  \int \left(3 r_a r_b - |\mathbf{r}|^2 \delta_{ab}\right) \boldsymbol{\rho}(\mathbf{r})\,d\mathbf{r},
  \label{eq:quadrupole}
\end{equation}
where $a, b \in \{x, y, z\}$. Both $\boldsymbol{\mu}$ and $Q_{ab}$ are therefore fully determined by $\boldsymbol{\rho}$ once the nuclear geometry is fixed.
 
\begin{table}[htbp]
\centering
\caption{Dipole and quadrupole MAE (OrbFlow vs.\ DFT) on held-out QM9. Relative MAE is normalized by the mean DFT-reference magnitude on the same test set.}
\label{tab:multipoles-from-density}
\small
\renewcommand{\arraystretch}{1.15}
\begin{tabular}{@{}l cc@{}}
\toprule
\textbf{Quantity} & \textbf{MAE} & \textbf{Rel.\ MAE [\%]} \\
\midrule
Dipole, $\mu_x, \mu_y, \mu_z$ & $0.026$\,D & $1.7$ \\
Dipole magnitude, $|\boldsymbol{\mu}|$ & $0.031$\,D & $1.0$ \\
Quadrupole, $Q_{ab}$ & $0.044$\,B & $2.5$ \\
Quadrupole norm, $\|Q\|_F$ & $0.066$\,B & $1.1$ \\
\bottomrule
\end{tabular}
\end{table}
 
Figure~\ref{fig:multipole-scatter} shows the per-molecule prediction accuracy for both moments across the held-out QM9 test set, confirming that the aggregate accuracy in Table~\ref{tab:multipoles-from-density} is uniform across molecules rather than driven by a favorable average over a subset of poorly-predicted outliers.
 
\begin{figure}[htbp]
\centering
\includegraphics[width=\linewidth]{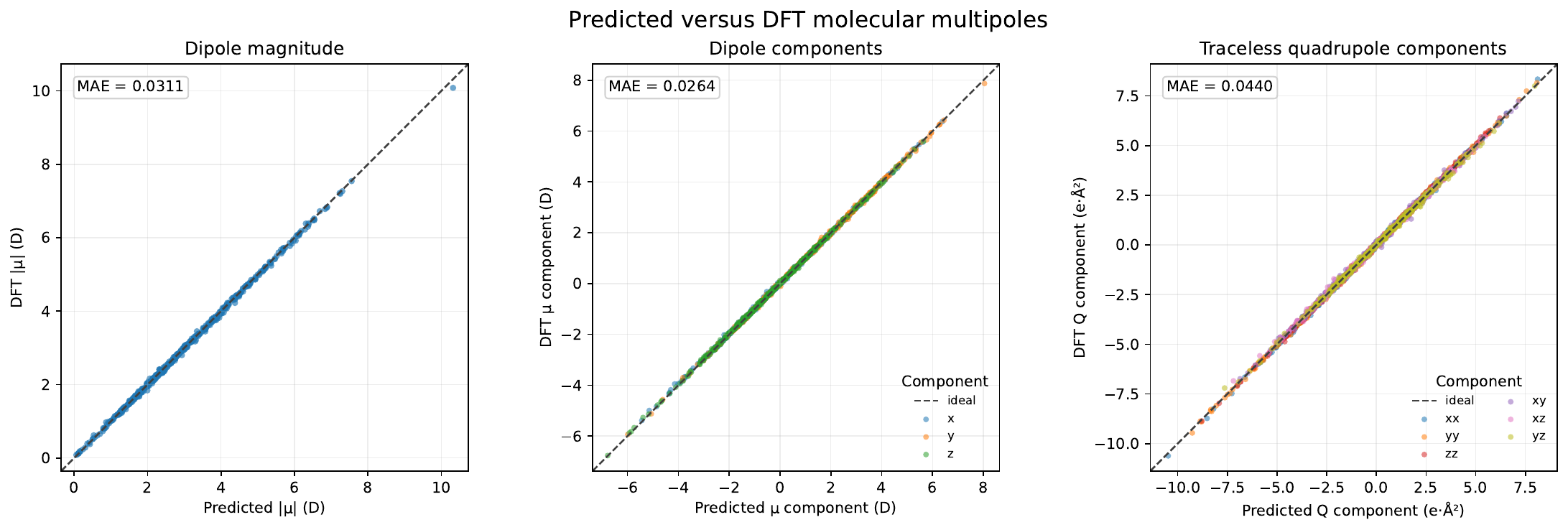}
\caption{Predicted (OrbFlow) vs.\ DFT-reference molecular multipoles on held-out QM9 (no SCF), computed via \eqref{eq:dipole} and \eqref{eq:quadrupole}. Left: dipole magnitude $|\boldsymbol{\mu}|$. Center: individual dipole components $\mu_x, \mu_y, \mu_z$. Right: individual traceless quadrupole components $Q_{ab}$. The dashed line is $y=x$. Points close to this line indicate accurate direct property prediction, and the reported MAE in each panel matches the corresponding row of Table~\ref{tab:multipoles-from-density}.}
\label{fig:multipole-scatter}
\end{figure}

\subsection{Efficiency Measurement Protocol}
\label{app:efficiency_details}

This section details the measurement protocol underlying the inference latency comparison in Sec.~\ref{sec:exp-efficiency}.

\textbf{Protocol.} All timings are measured on a single NVIDIA H200 GPU with batch size 1. We report the mean over the six MD test molecules (Sec.~\ref{sec:exp-md}) and over three independent timing repeats of the same checkpoints, where each repeat covers the full test split after a short CUDA warm-up of three graphs. The network cost is the wall-clock time to produce final GTO coefficients from the molecular graph, which for OrbFlow includes the full $K$-step Euler rollout in \eqref{eq:generalized_euler}. The density synthesis cost is the cost of the operator $\mathcal{O}(\cdot)$ in \eqref{eq:synth} mapping coefficients to the real-space grid. OrbFlow evaluates this operator with a fused Triton reduction of the same GTO contraction used by SCDP.

\begin{table}[htbp]
\centering
\caption{Per-molecule inference latency on the MD benchmark, averaged across the six molecules and three independent H200 timing repeats (mean $\pm$ std).}
\label{tab:efficiency-comparison}
\small
\renewcommand{\arraystretch}{1.1}
\begin{tabular}{@{}lccc@{}}
\toprule
\textbf{Model} & \textbf{Network Cost [s]} & \textbf{Density Synthesis Cost [s]} & \textbf{Total [s]} \\
\midrule
ELECTRA & $0.058 \pm 0.003$ & $0.0025 \pm 0.0002$ & $0.061 \pm 0.003$ \\
SCDP & $0.0265 \pm 0.0004$ & $0.0150 \pm 0.0001$ & $0.0414 \pm 0.0005$ \\
OrbFlow (ours) & $0.054 \pm 0.001$ & $0.0126 \pm 0.0001$ & ${0.066 \pm 0.001}$ \\
\bottomrule
\end{tabular}
\end{table}

\textbf{Scaling with $K$.} Two patterns in Table~\ref{tab:efficiency-comparison} are worth expanding on. First, OrbFlow's network cost scales close to linearly with $K$. Its $0.054$s cost at two integration steps is roughly twice SCDP's single-shot cost of $0.026$s, consistent with a per-step cost close to a single regression forward pass, around $0.027$s per step. OrbFlow's inference cost is therefore predictable and directly controllable through the choice of $K$ (Sec.~\ref{sec:ablations}) rather than compounding unpredictably. Second, the density synthesis cost differs substantially between ELECTRA ($0.003$s) and the two GTO-based models, SCDP ($0.015$s) and OrbFlow ($0.013$s). This gap reflects the underlying density representation, a 3D Gaussian mixture evaluated with a fused KeOps reduction for ELECTRA against the auxiliary GTO basis expansion shared by SCDP and OrbFlow, rather than the flow matching formulation itself.

\section{Extended Ablation Results}
\label{app:additional_ablations}
 
This appendix expands on the ablations summarized in Sec.~\ref{sec:ablations}, covering the deterministic two-step baseline, training compute cost, the training-time rollout horizon, and the inference-time integration budget.
 
\subsection{Deterministic Two-Step Baseline}
\label{app:det-baseline}
 
The deterministic baseline of Sec.~\ref{sec:ablations} uses the same EquiformerV3 configuration as the QM9 models (Table~\ref{tab:hparams}), including the equivariant linear map that injects the coefficient state at each layer, and is trained and evaluated on the same 10\% QM9 ablation subset as the other ablation rows. The first evaluation receives an all-zero coefficient state and the second receives the output of the first. No noise is sampled, and the continuous time embedding is replaced by a learned two-valued step-index embedding, so that the baseline retains the ability to behave differently on its first and second evaluation. This isolates the stochastic prior and the bridge interpolation, the components specific to the flow formulation, rather than conflating them with the loss of step conditioning. The chain is trained end to end with the real-space MAE of~\eqref{eq:Lint} under the same step budget as the corresponding OrbFlow run, with identical optimizer, learning-rate schedule, EMA, probe subsampling, and checkpoint selection. 
 
\subsection{Curriculum Variants and Training Compute}
\label{app:curriculum_compute}
 
Table~\ref{tab:ablation_curriculum_full} collects every configuration discussed in Sec.~\ref{sec:ablations} together with its total training cost. The curriculum is compute-efficient in practice because the Phase-1 warm start is shared across any number of downstream rollout horizons. Of the 74 GPU-h total cost for the full $K{=}2$ curriculum, only 39 GPU-h come from the Phase-2 fine-tuning stage itself, with the remaining 35 GPU-h of Phase-1 cost shared with $K{=}3$. That incremental fine-tuning cost is roughly half that of training the same $K{=}2$ objective from scratch (81 GPU-h), and the saving widens once the shared Phase-1 cost is amortized across rollout horizons.
 
\begin{table}[htbp]
    \centering
    \caption{All ablation configurations on the 10\% QM9 subset, with total training compute. NMAE is mean $\pm$ std over three seeds. Rows without a compute entry reuse the shared Phase-1 checkpoint and are reported for accuracy only. }
    \label{tab:ablation_curriculum_full}
    \small
    \begin{tabular}{@{}lcc@{}}
        \toprule
        \textbf{Configuration} & \textbf{NMAE (\%)} $\downarrow$ & \textbf{GPU-h} \\
        \midrule
        Direct regression (single pass) & $0.280 \pm 0.003$ & -- \\
        Deterministic two-step & $ 0.268\pm 0.003$  & -- \\
        \midrule
        Phase 1 only & $0.295 \pm 0.016$ & 35 \\
        Phase 2 only, $K{=}2$ (no warm start) & $0.250 \pm 0.027$ & 81 \\
        Joint single-stage, $K{=}2$ & $0.241 \pm 0.002$ & 182 \\
        \textbf{Full curriculum, $K{=}2$ (Ours)} & $\mathbf{0.240 \pm 0.005}$ & \textbf{74} \\
        \midrule
        Full curriculum, $K{=}1$ & $0.296 \pm 0.005$ & -- \\
        Full curriculum, $K{=}3$ & $0.240 \pm 0.002$ & --  \\
        \bottomrule
    \end{tabular}
\end{table}
 
\subsection{Training-Time Rollout Horizon}
\label{app:rollout_horizon}
 
As shown in Table~\ref{tab:ablation_curriculum_full}, holding the Phase-1 warm start fixed, we vary the rollout budget $K$ used during Phase 2. $K{=}1$ reaches $0.296 \pm 0.005\%$, statistically indistinguishable from Phase-1-only training at $0.295 \pm 0.016\%$. This follows from the endpoint parameterization. A single Euler step starting at $t_0 = 0$ evaluates $f_\theta$ at exactly the point $\mathcal{L}_{\mathrm{stage1}}$ already supervises, so fine-tuning restricted to this boundary case adds essentially no signal beyond Phase 1.
 
$K{=}2$ and $K{=}3$ converge to the same accuracy, $0.240 \pm 0.005\%$ and $0.240 \pm 0.002\%$. Total training compute nevertheless grows monotonically with $K$, since each additional Euler step adds one more sequential network evaluation per training example during Phase 2, together with the corresponding backward pass. We therefore adopt $K{=}2$ as the smallest budget at which rollout supervision is informative and past which additional steps yield no measurable benefit at increasing cost.
 
\subsection{Inference-Time Integration Budget}
\label{app:inference_budget}
 
Table~\ref{tab:ablation_nfe} reports accuracy and relative cost as a function of the number of Euler steps used at inference for our $K{=}2$-trained model, extending to budgets never seen during training. Accuracy is best exactly at the train-matched budget of two steps. Using more steps than trained does not recover accuracy but degrades it, from $0.240\%$ to $0.539\%$ at four steps, after which the error plateaus in the $0.52$--$0.54\%$ range while relative cost grows to $3.98\times$ at sixteen steps. The train-matched setting is therefore both the accuracy-optimal and the most compute-efficient inference choice.
 
\begin{table}[htbp]
    \centering
    \caption{QM9 NMAE as a function of the number of Euler integration steps used at inference, for a model trained with $K{=}2$. NMAE is mean $\pm$ std over three training seeds (best-val checkpoint). Relative wall-clock inference cost is normalized to the train-matched budget ($t_K/t_2$), averaged over seeds.}
    \label{tab:ablation_nfe}
    \small
    \begin{tabular}{lcccc}
        \toprule
        \textbf{Inference Steps} & \textbf{2 (train-matched)} & \textbf{4} & \textbf{8} & \textbf{16} \\
        \midrule
        \textbf{NMAE (\%)} $\downarrow$ & $\mathbf{0.240 \pm 0.005}$ & $0.539 \pm 0.033$ & $0.519 \pm 0.044$ & $0.516 \pm 0.053$ \\
        \textbf{Relative Inference Cost} & $1\times$ & $1.36\times$ & $2.18\times$ & $3.98\times$ \\
        \bottomrule
    \end{tabular}
\end{table}
 
This lack of generalization beyond the trained horizon follows directly from how Phase 2 supervises the model, and is worth making explicit. Under the endpoint parameterization of~\eqref{eq:v-from-endpoint}, $\mathbf{v}_\theta$ is derived algebraically from a predicted clean endpoint $f_\theta(\mathcal{G}, \mathbf{x}_t, t)$, and Phase 2 supervises that endpoint prediction only at the two discretization points $t \in \{0, 0.5\}$ visited by the $K{=}2$ rollout. The network therefore receives no training signal at the intermediate timesteps an integrator with $K > 2$ would query, nor at the coarser single interpolation step $K{=}1$ would require. The plateau in Table~\ref{tab:ablation_nfe} is consistent with this account. Rather than learning a smooth field valid under arbitrary discretizations, Phase 2 sharpens $f_\theta$ into an accurate two-point discrete map, and additional Euler steps subdivide the trajectory into intervals the endpoint predictor was never asked to be accurate on. Removing this restriction, for instance by conditioning on the step size as in shortcut models~\citep{frans2025one}, is the clearest direction for future work.

\end{document}